\documentclass[10pt]{article}

\usepackage[
	paperwidth		=8.5in,
	paperheight		=11in,
	top				=1.25in, 
	bottom			=1.25in, 
	left			=1.0in, 
	right			=1.0in
	]{geometry}
	
\usepackage{type1ec}
\usepackage[T1]{fontenc}
\usepackage[english]{babel}
\usepackage{lmodern}
\usepackage{microtype}
\usepackage{natbib}
\usepackage[shortlabels]{enumitem}
\usepackage{mdwlist}
\usepackage{wrapfig}
\usepackage{subcaption}
\usepackage{lipsum}

\usepackage{amssymb}
\usepackage{amsmath}
\usepackage{amsthm}

\usepackage[mathscr]{eucal}

\usepackage{dsfont}

\usepackage[dvipsnames]{xcolor}	

\usepackage[plainpages=false, pdfpagelabels]{hyperref}

\usepackage{mathtools}                     

\usepackage{ifpdf}
\ifpdf
	\usepackage[pdftex]{graphicx}
	\usepackage{epstopdf}
\else
	\usepackage{graphicx}
\fi

\usepackage[frenchmath]{mathastext}
\Mathastext

\let\oldtitle\title
\renewcommand{\title}[1]{\oldtitle {\scshape #1}}

\usepackage{sectsty}
\allsectionsfont{\mdseries\scshape}

\mathtoolsset{showonlyrefs=true}

\hypersetup{
	colorlinks   = true,
	citecolor    = RoyalBlue, 
	linkcolor    = RubineRed, 
	urlcolor     = Turquoise
	}

\allowdisplaybreaks               

\numberwithin{equation}{section}

\newtheoremstyle{plain}
  {}   
  {}   
  {\itshape}  
  {}       
  {\mdseries\scshape} 
  {.}         
  { } 
  {\thmname{#1}\thmnumber{ #2}\ifx#3\empty\else\ (#3)\fi}
\theoremstyle{plain}
\newtheorem{theorem}{\underline{Theorem}} 
        
\newtheorem{proposition}[theorem]{\underline{Proposition}}
\newtheorem{corollary}[theorem]{\underline{Corollary}}

\newtheoremstyle{definition}
  {}   
  {}   
  {}  
  {}       
  {\mdseries\scshape} 
  {.}         
  { } 
  {\thmname{#1}\thmnumber{ #2}\ifx#3\empty\else\ (#3)\fi}
\theoremstyle{definition}
\newtheorem{definition}[theorem]{\underline{Definition}}
\newtheorem{example}[theorem]{\underline{Example}}
\newtheorem{remark}[theorem]{\underline{Remark}}

\expandafter\let\expandafter\oldproof\csname\string\proof\endcsname
\let\oldendproof\endproof
\renewenvironment{proof}[1][\proofname]{%
  \oldproof[\scshape\underline{#1}]%
}{\oldendproof}

\numberwithin{theorem}{section}

\newcommand{\<}{\left\langle}
\renewcommand{\>}{\right\rangle}
\renewcommand{\(}{\left(}
\renewcommand{\)}{\right)}
\renewcommand{\[}{\left[}
\renewcommand{\]}{\right]}

\newcommand\Eb{\mathds{E}}
\newcommand\Fb{\mathds{F}}

\newcommand\Pb{\mathds{P}}

\newcommand\Rb{\mathds{R}}

\newcommand\Ac{\mathscr{A}}

\newcommand\Fc{\mathscr{F}}

\newcommand\Hc{\mathscr{H}}
\newcommand\Lc{\mathscr{L}}

\newcommand\Nc{\mathscr{N}}

\newcommand\Pc{\mathscr{P}}

\newcommand\eps{\varepsilon}
\newcommand\om{\omega}
\newcommand\Om{\Omega}
\newcommand\sig{\sigma}

\newcommand\gam{\gamma}

\newcommand\del{\delta}

\newcommand\ab{\overline{a}}
\newcommand\bb{\overline{b}}

\newcommand\hb{\overline{h}}

\newcommand\yv{\mathbf{y}}

\newcommand\Cv{\mathbf{C}}
\newcommand\mv{\mathbf{m}}
\newcommand\av{\mathbf{a}}

\newcommand\Lch{\widehat{\Lc}}
\newcommand\Sh{\widehat{S}}

\newcommand\Gamh{\widehat{\Gamma}}

\renewcommand\d{\partial}

\newcommand\dd{\mathrm{d}}
\newcommand\ee{\mathrm{e}}

\newcommand{\nuba}[1]{\overline{\nu_{#1}^\ast}}

\newcommand{\blu}[1]{\textcolor{blue}{#1}}

\newcommand{\ora}[1]{\textbf{\textcolor[cmyk]{0,.61,.97,0}{#1}}}

\begin{document}

\title{Optimal liquidation under stochastic price impact}

\author{
Weston Barger
\thanks{Department of Applied Mathematics, University of Washington.  \textbf{e-mail}: \url{wdbarger@uw.edu}}
\and
Matthew Lorig
\thanks{Department of Applied Mathematics, University of Washington.  \textbf{e-mail}: \url{mlorig@uw.edu}}
}

\date{This version: \today}

\maketitle


%
%

\begin{abstract}
We assume a continuous-time price impact model similar to Almgren-Chriss but with the added assumption that the price impact parameters are stochastic processes modeled as correlated scalar Markov diffusions. In this setting, we develop trading strategies for a trader who desires to liquidate his inventory but faces price impact as a result of his trading. For a fixed trading horizon, we perform coefficient expansion on the Hamilton-Jacobi-Bellman equation associated with the trader's value function. The coefficient expansion yields a sequence of partial differential equations that  we solve to give closed-form approximations to the value function and optimal liquidation strategy. We examine some special cases of the optimal liquidation problem and give financial interpretations of the approximate liquidation strategies in these cases. Finally, we provide numerical examples to demonstrate the effectiveness of the approximations.  
\end{abstract}

\noindent

\section{Introduction}
\label{sec:intro}

When institutional traders execute large market orders, they are faced with transactional frictions. Direct frictions, such as exchange and brokerage fees, are known in advance and can be incorporated into a trading strategy. Traders also incur indirect costs, and such costs are, in general, unknown in advance and may be difficult to quantify even after the trading is complete. The opportunity cost that arises from waiting to execute trades and the price impact that results from trading are both examples of indirect costs. Price impact typically affects traders adversely. Selling an asset puts downward pressure on the price thereby lowering revenues while purchasing an asset pushes its price upward, resulting in higher costs. We focus on price impact costs in this paper. Specifically, we examine how a trader should optimally liquidate a large position in a market in which price impact is stochastic.

The optimal liquidation problem under price impact has been studied extensively in the literature. \cite{bertsimas1998optimal} use a linear price impact model and solve a discrete optimal control problem to minimize expected trading costs. \cite{almgren1999value, almgren2001optimal, huberman2005optimal} also use a linear price impact model but consider the variance in trading costs. \cite{almgren2003optimal} employs nonlinear impact functions and discusses the continuous-time limit of the models in \cite{almgren1999value, almgren2001optimal} in more detail. \cite{almgren2012optimal} considers optimal liquidation in a market with stochastic liquidity and stochastic volatility. \cite{NBERw11444} include price impact by modeling the limit order book directly (see also the published version, \cite{obizhaeva2013optimal}). \cite{alfonsi2010optimal} extend the work of \cite{NBERw11444} to allow for general limit order book shapes. \cite{cartea2016incorporating} use a continuous time linear impact model and incorporate stochastic order flow. For an overview of continuous-time price impact models, see \cite{cartea2015algorithmic} and the references therein. 

In this paper, we assume a continuous-time price impact model where the price impact parameters are stochastic. Specifically, the temporary and permanent price impact parameters are modeled as scalar Markov diffusions. We allow the temporary and permanent price impact parameters to be correlated, as empirical evidence suggests they are (see \cite{cartea2016incorporating}). In this setting, we define a trader's value function and formulate the associated Hamilton-Jacobi-Bellman (HJB) partial differential equation (PDE). We find an approximate solution of the HJB equation by applying coefficient expansion techniques that were first developed for one-dimensional linear parabolic PDEs in \cite{pagliarani2011analytical} and later extended to $d$-dimensions in \cite{lorig-pagliarani-pascucci-4} and nonlinear problems in \cite{lorig2016portfolio,lorig-4}. This, in turn, yields approximations to the associated optimal trading strategy. The resulting optimal strategy approximations are explicit and do not require numerical integration. 

The zeroth order approximation to the optimal strategy can be interpreted as an Almgren-Chriss strategy for which the price impact parameters are recalibrated in continuous time. Successive terms in higher-order approximations can therefore be viewed as corrections to the strategy of Almgren and Chriss. Higher-order strategy approximations are influenced by the geometry of the stochastic differential equation (SDE) coefficients modeling the impact parameter's dynamics, allowing traders to take advantage of periods of relatively high or low price impact. 

The rest of the paper is organized as follows. In Section \ref{sec:model}, we state our modeling assumptions, define the trader's value function, and provide the associated HJB equation. In Section \ref{sec:asymptotics}, we develop an asymptotic expansion for solutions of the HJB equation. This expansion leads to a sequence of PDEs, which we solve recursively in Section \ref{sec:pde.solution}. The solution of these PDEs allows us to construct approximations to the optimal liquidation strategy. We discuss limiting cases of the optimal strategy approximations in Section \ref{sec:limstrats}. In Section \ref{sec:simulation}, we demonstrate the effectiveness of the approximate optimal strategies by performing a Monte Carlo study. Some concluding remarks are offered in Section \ref{sec:conclusion}. 

%
%
\section{Market model and trader's value function}
\label{sec:model}

To begin, we fix a trading horizon $T > 0$ and filtered probability space $(\Om, \Fc, \Fb=(\Fc_t)_{ 0\leq t \leq T}, \Pb)$. We suppose an institutional trader holds $Q_0 > 0$ shares of a stock $S$ that he wishes to liquidate. The trader does not post limit orders but trades exclusively via market orders. He must choose the speed at which he sends market orders with the aim of liquidating all $Q_0$ shares by the end of a trading horizon $T$. We assume that the trader trades in continuous time, and we denote by $\nu = ( \nu_t )_{0\leq t \leq T}$ the rate at which the trader sends market orders (i.e., the liquidation speed). The inventory $Q^\nu = (Q_t^\nu)_{0\leq t \leq T}$ depends on the the trading strategy $\nu$ and is given by
\begin{align}
\dd Q_t^\nu 
	&= -\nu_t \dd t. \label{dynamics:Q}
\end{align}
A positive trading rate $\nu_t > 0$ at a time $t$ corresponds to selling shares of $S$, and a negative rate $\nu_t  < 0$ corresponds to buying shares of $S$. Although we shall restrict ourselves to the liquidation problem in this paper, we mention that the set up for the acquisition problem is similar. The trader wishes to choose $\nu$ such that he minimizes the indirect costs he incurs as a result of his trading. We incorporate price impact in the model by explicitly including temporary and permanent price impact parameters. 

%
%

\subsection{Permanent price impact}
\label{subsec:pimpact}

We assume that when the trader sends market orders there is a permanent impact on the midprice of the stock. Sell orders put downward pressure on the midprice of the stock, and, conversely, buy orders put upward pressure on the midprice. For example, suppose that a trader submits a large sell order for $S$, and suppose further that other traders on the market have similar signals and also post market orders to sell. Liquidity providers fill the gap in the book by posting limit orders to sell at lower prices, thus moving the midprice down. 

We model the midprice of the stock as a stochastic process $S^\nu = ( S_t^\nu)_{0 \leq t \leq T}$ with the dynamics
\begin{align}
\dd S_t^\nu
	&= -g(b_t) \nu_t\, \dd t + \sig\, \dd W_t,  \label{dynamics:S}
\end{align}
where the constant $\sig > 0$ is positive, the function $g$ is continuous and real-valued, the Markov diffusion $b = (b_t)_{0 \leq t \leq T}$ has the dynamics
\begin{align}
\dd b_t
	&= \eta(b_t)\,  \dd t + \psi(b_t)\, \dd B_t^{(1)}, \label{dynamics:b}
\end{align}
and the standard Brownian motions $W = (W_t)_{0 \leq t \leq T}$ and $B^{(1)} = (B_t^{(1)})_{0 \leq t \leq T}$ are uncorrelated.

We think of the Brownian motion $W$ as market noise due to the reshuffling of limit orders. Permanent price impact is modeled by the process $g(b)=(g(b_t))_{0 \leq t \leq T }$, the magnitude of which corresponds the severity of the impact. We require that $g(b_t) \geq 0$ for all $0 \leq t \leq T$ because a negative permanent price impact would imply that selling shares of an asset would push the midprice upwards, which is unrealistic. Asymptotic expansions are performed in Section \ref{sec:asymptotics} for general $g$, but we give the linear model $g(b) = b$ special attention in Section \ref{sec:simulation}.

%
%

\subsection{Temporary impact}
\label{subsec:timpact}

In addition to permanent price impact, the trader also faces a temporary price impact. Temporary price impact is the cost directly associated with each trade, and, unlike the permanent impact, temporary impact does not carry over into subsequent trades. Temporary impact can be understood as follows: the number of shares available at the best bid is limited, and if the trader's market order is large enough then the trader walks the book (i.e., depletes the outstanding limit orders nearest the midprice). We include temporary price impact in the model by defining the execution price $\Sh^\nu = (\Sh_t^\nu)_{0 \leq t \leq T}$ of the asset to be
\begin{align}
\Sh_t^\nu
	&= S_t^\nu -  f(a_t)\nu_t, \label{dynamics:Shat}
\end{align}
where the function $f$ is continuous and real-valued and the Markov diffusion $a = (a_t)_{0 \leq t \leq T}$ has the dynamics
\begin{align}
\dd a_t
	&= \mu(a_t)\,  \dd t + \om(a_t)\, \dd B_t^{(2)} . \label{dynamics:a}
\end{align}
Here, the standard Brownian motion $B^{(2)} = (B_t^{(2)})_{0\leq t \leq T}$ is uncorrelated with the Brownian motion $W$ that drives the midprice $\dd \< W, B^{(2)}\> = 0$ but is correlated with the Brownian motion $B^{(1)}$ that drives the permanent price impact $\dd \< B^{(1)}, B^{(2)} \> = \rho \, \dd t$ where $\rho \in [-1,1]$.  Taking the temporary price impact to be a stochastic process allows us to incorporate stochastic liquidity into our model.

We require that $f(a_t) > 0$ for all $0 \leq t \leq T$ to reflect the fact that traders are not compensated for posting market orders. Assuming the trader wishes to minimize the cost associated with temporary price impact, if the temporary price impact ever reached zero, the trader would liquidate his entire inventory immediately resulting in blowup in the optimal strategy. The magnitude of the process $f(a_t)$ describes the severity of the temporary price impact. We note that \cite{cartea2016incorporating} suggests that temporary and permanent price impact are correlated, so we allow temporary and permanent price impact to be correlated with parameter $\rho$.  Generally, temporary and permanent price impact are positively correlated, although our model does not require that. 

The above framework stipulates that the temporary impact is only felt by the trader who initiates the market order. Furthermore, the limit order book rebalances infinitely fast to the state before the arrival of the market order. This assumption is known as order book resilience. See \cite{alfonsi2010optimal}, \cite{almgren2003optimal}, \cite{gatheral2012transient}, \cite{kharroubi2010optimal}, and \cite{schied2013robust} for further study and relaxations of the resilience assumption. Asymptotic expansions are performed in Section \ref{sec:asymptotics} for general $f$, but we give special attention to the linear case $f(a) = a$ in Section \ref{sec:simulation}. 

In the framework described above, one easily derives that the trader's cash position $X^\nu = (X_t^\nu)_{0 \leq t \leq T}$ is given by 
\begin{align}
\dd X_t^\nu	
	&= \nu_t \Sh_t^\nu \dd t
	= \nu_t\( S_t^\nu - f(a_t)\nu_t \) \dd t. \label{dynamics:X}
\end{align}

%
%

\subsection{Trader's value function}
\label{sec:problem}

We consider a trader who wishes to liquidate $Q_0$ shares of $S^\nu$ under the model described in Section \ref{sec:model}. We assume that the trader wishes to maximize his expected cash at the terminal time $T$ subject to penalties for holding inventory. For a given trading strategy $\nu$, we define the trader's performance criteria $H^\nu$ to be
\begin{align}
H^\nu(t,x,s,q,a,b)
	:= \Eb_{t,x,s,q,a,b}\[ X_T^\nu + Q_T^\nu \(S_T^\nu - \kappa Q_T^\nu\) - \varphi \int_t^T \dd s ( Q_s^\nu)^2 \], \label{eq:perf.criteria}
\end{align}
where the constants $\kappa > 0$ and $\varphi > 0$ are positive and $\Eb_{t,x,s,q,a,b}$ is shorthand for expectation conditioned on $(X_t^\nu, S_t^\nu , Q_t^\nu, a_t, b_t) = (x,s,q,a,b)$. From left to right, the following three terms are present in the trader's performance criteria \eqref{eq:perf.criteria}: terminal cash, the proceeds of liquidating the remaining shares at the terminal time $T$, and an integral term penalizing the holding of inventory. The proceeds from liquidation at time $T$ are subject to temporary price impact, which is incorporated through the parameter $\kappa$. The third term $\varphi \int_t^T \dd s ( Q_s^\nu)^2$ imposes a running penalty for holding inventory. When $\varphi$ is large, optimal strategies will trade quickly at the beginning of the trading horizon rather than face holding large inventories. In \cite{cartea2014algorithmic}, the authors show that including the inventory penalty term is equivalent to the trader considering alternate models with stochastic drifts but penalizing models that are far from the reference model in the sense of relative entropy. In that context, larger values of $\varphi$ correspond to an trader who is less confident about the drift of the $S^\nu$. The authors of \cite{cartea2015risk} introduce the inventory penalty term heuristically and justify it by showing that it is proportional to the variance of the book value of the inventory over the trading horizon. 

The trader's value function is given by
\begin{align}
H(t,x,s,q,a,b)
	&= \sup_{\nu \in \mathcal{A}} H^\nu(t,x,s,q,a,b), \label{eq:value.function}
\end{align}
where $\mathcal A= \{ \nu \mid \nu \text{ is } \Fc_t \text{ adapted and } \int_t^T\dd s |\nu_s| < \infty, \Pb - \text{a.s.}\}$ is the set of admissible strategies.

%
%

\subsection{The Hamilton-Jacobi-Bellman equation}
\label{sec:dynamic.programming}

In this section, we give the HJB equation associated with the value function $H$. Let $\Hc^\nu$ denote the infinitesimal generator for the process $(X^\nu, S^\nu, Q^\nu, a,b)$ with $\nu$ fixed. Explicitly,
\begin{align}
\Hc^\nu
	&= \frac 12 \sig^2 \d_s^2 - g(b) \nu \d_s - \nu \d_q + \nu\( s - f(a)\nu\)  \d_x   \\ & \qquad
	  + \frac 12 \om^2(a)  \d_a^2 + \rho\, \om(a) \psi(b) \d_a\d_b + \frac 12 \psi^2(b)  \d_b^2 + \mu(a) \d_a + \eta(b)  \d_b.\label{eq:inf.gen}
\end{align}
As we shall see in later in this section, it is convenient to write $\Hc^\nu$ as the sum of two operators
\begin{align}
\Hc^{\nu}
	&= \Ac^\nu + \Lc, &
\Ac^{\nu}
	&:= \frac 12 \sig^2 \d_s^2 - g(b) \nu \d_s - \nu \d_q + \nu\( s - f(a) \nu \) \d_x, \label{eq:Anu} \\
& & \Lc
	&:= \frac 12 \om^2(a)\d_a^2 + \rho\, \om(a) \psi(b) \d_a\d_b + \frac 12 \psi^2(b)\d_b^2 + \mu(a) \d_a + \eta(b) \d_b. \label{eq:L}
\end{align}
The operator $\Lc$ is the infinitesimal generator of the process $(a,b)$ and the operator $\Ac^\nu$ is the infinitesimal generator of $(X^\nu, S^\nu, Q^\nu)$ with $(\nu,a,b)$ fixed. When the process $(a,b)$ is constant (i.e., $(\mu,\eta, \om, \psi) = (0,0,0,0)$), we have $\Hc^\nu = \Ac^\nu$, and the model reduces to the continuous-time Almgren-Chriss model. 

The HJB equation associated with the trader's value function $H$ is  
\begin{align}
\(\d_t + \Lc\) H + \sup_{\nu} \( \Ac^\nu H - \varphi q^2\)
	&= 0, &
H(T,x,s,q,\cdot,\cdot)
	= x + q\(s - \kappa q\), \label{pde:hjb}
\end{align}
where $\Ac^\nu$ and $\Lc$ are given by \eqref{eq:Anu} and \eqref{eq:L}, respectively. We assume that \eqref{pde:hjb} admits a unique classical solution which coincides with the trader's value function (see \cite{pham2009continuous}). 

Following \cite{cartea2016incorporating}, we make the following ansatz
\begin{align}
H(t,x,s,q,a,b) 
	&= x + q s + q^2 h(t,a,b), \label{eq:ansatz1}
\end{align}
for some function $h$ to be determined. We refer to $h$ as the \textit{transformed value function}. Inserting \eqref{eq:ansatz1} into \eqref{pde:hjb} yields the following PDE problem for $h$:
\begin{align}
q^2(\d_t + \Lc)h + \sup_{\nu} \(-q (2 \nu  h +\nu  g +q \varphi )-\nu ^2 f\)
	&= 0, &
h(T,\cdot,\cdot)
	&= -\kappa . \label{pde:hjb.lite}
\end{align}
The optimal strategy $\nu^\ast$, obtained by maximizing the supremum in \eqref{pde:hjb.lite}, is given in feedback form as
\begin{align}
\nu^\ast(t,q,a,b)
	&=- \( \frac{ g(b) + 2 h(t,a,b)}{2 f(a)}\)q. \label{eq:nu} 
\end{align}
Inserting \eqref{eq:nu} into \eqref{pde:hjb.lite} we obtain
\begin{align}
0
	&= (\d_t + \Lc) h + \Nc(h) - \varphi, &
h(T,\cdot,\cdot,)
	&= - \kappa , \label{pde:h} \\
\Nc(h)
	&= \frac{1}{f} h^2 + \frac{g}{f} h + \frac{g^2}{4 f}. \label{eq:N}
\end{align}
Note that we have reduced the HJB equation \eqref{pde:hjb} to a PDE that involves only three variables: $(t,a,b)$.

%
%

\section{Asymptotics}
\label{sec:asymptotics}

For general $(f,g,\om,\mu,\psi,\eta,\rho)$, there is no closed-form solution to \eqref{pde:h}. In this section, we develop a formal asymptotic expansion for the transformed value function $h$ and the corresponding optimal execution strategy $\nu^\ast$ by performing polynomial expansions on the coefficients of \eqref{pde:h}. The authors of \cite{lorig-pagliarani-pascucci-4} use this approach for the European option pricing problem in a general local-stochastic volatility setting. One key difference here is that, unlike classical option pricing PDEs, which are linear, the PDE \eqref{pde:h} is nonlinear. Our approach is similar to that of \cite{lorig2016portfolio} and \cite{lorig-4}, who apply the polynomial coefficient expansion method to the Merton problem and indifference pricing problem, both of which are nonlinear.

%
%

\subsection{Coefficient Taylor series expansions}
\label{sec:pde.asymptotics}

For the sake of simplicity, we assume in the following formal computations that the coefficients of \eqref{pde:h} are analytic. We shall see later that the $N$th-order approximation we obtain for $h$ and $\nu^\ast$ require only that the coefficients of \eqref{pde:h} belong to $C^N(D)$ where $D$ is some open set in $\Rb^2$.

Let $\chi$ be a placeholder for any of the coefficients appearing in PDE \eqref{pde:h}
\begin{align}
\chi   
    \in \{ \tfrac 12\om^2 ,\rho\om\psi,\tfrac 12 \psi^2,\mu,\eta,f^{-1},f^{-1} g,4^{-1} f^{-1} g^2\}, \label{eq:chi}
\end{align}
and fix a point $(\ab, \bb) \in \Rb^2$.  For any $\eps \in [0,1]$, we define
\begin{align}
\chi^\eps (a,b)
    &:= \chi(\ab + \eps( a - \ab), \bb + \eps (b - \bb)). \label{eq:chi.eps}
\end{align}
Formally, Taylor expanding $\chi^\eps$ in $\eps$ about the point $\eps = 0$ gives
\begin{align}
\chi^\eps(a,b)
	&:= \sum_{n=0}^\infty \eps^n \chi_n(a,b), &
\eps 
	&\in  [0,1], \label{eq:chi.expand}\\
\chi_n(a,b)
	&:= \sum_{k=0}^n \chi_{n-k,k} \cdot (a - \ab)^{n-k} (b - \bb)^k, &
\chi_{n-k,k}
	&:= \frac{1}{(n-k)! k!} \d_a^{n-k} \d_b^{k} \chi (\ab,\bb). \label{eq:chi.n}
\end{align} 
In particular, evaluating \eqref{eq:chi.expand} at $\eps = 1$ yields the Taylor series expansion of $\chi$ about the point $(\ab,\bb)$. Consider now the family of PDEs indexed by $\eps$:
\begin{align}
( \d_t + \Lc^\eps) h^\eps + \Nc^\eps(h^\eps) - \varphi
	&= 0, &
h^\eps(T,\cdot,\cdot,\cdot)
	&= - \kappa, &
\eps
	&\in [0,1], \label{pde:h.eps}
\end{align}
where $\Lc^\eps$ and $\Nc^\eps(\cdot)$ are the operators obtained by replacing the coefficients of $\Lc$ and $\Nc(\cdot)$ in \eqref{eq:L} and \eqref{eq:N}, respectively, with their $\eps$-counterparts. Explicitly, we make the replacements
\begin{align}
\{\tfrac 12\om^2 ,\rho\om\psi,\tfrac 12\psi^2,\mu,\eta,f^{-1},f^{-1} g,4^{-1} f^{-1} g^2 \} 
	\mapsto \{(\tfrac 12\om^2)^\eps ,(\rho\om\psi)^\eps,(\tfrac 12\psi^2)^\eps,\mu^\eps,\eta^\eps,(f^{-1})^\eps,(f^{-1} g)^\eps,(4^{-1} f^{-1} g^2)^\eps\}
	\label{eq:epsmapping}
\end{align}
in \eqref{pde:h} to obtain \eqref{pde:h.eps}. Using \eqref{eq:chi.expand}, the linear operator $\Lc^\eps$ in the PDE \eqref{pde:h.eps} can be written as
\begin{align}
\Lc^\eps
	= \sum_{n=0}^\infty \eps^n \Lc_n, \label{eq:L.eps}
\end{align}
where we have defined
\begin{align}
\Lc_n
	:= (\tfrac 12 \om^2)_n \d_a^2 + (\rho \om \psi)_n \d_{ab}^2 + (\tfrac 12 \psi^2)_n \d_b^2 + \mu_n \d_a + \eta_n \d_b, \label{eq:Ln}
\end{align}
and the subscript notation $\chi_n$ is as described in \eqref{eq:chi.n}. The expansion of the nonlinear operator $\Nc^\eps$ is more involved, and we handle it below. 

We construct an expansion for the function $h^\eps$, the solution to the PDE \eqref{pde:h.eps}, as the power series in $\eps$
\begin{align}
h^\eps (t,a,b)
	&= \sum_{n=0}^\infty \eps^n h_n(t,a,b), &
\eps
	&\in [0,1]. \label{eq:h.eps}
\end{align}
Here, the sequence of functions $(h_n)_{n=0}^\infty$ are not polynomials in $(a,b)$ but rather functions to be determined which, in particular, are independent of $\eps$. We shall eventually construct the asymptotic approximation to the transformed value function $h$ by truncating \eqref{eq:h.eps} for some $n = N$ and setting $\eps = 1$. 

We insert \eqref{eq:L.eps} and \eqref{eq:h.eps} into \eqref{pde:h.eps}, expand the terms in $\Nc^\eps(h^\eps)$ in powers of $\eps$, and collect terms of like order in $\eps$. As the equality in \eqref{pde:h.eps} holds for all for all $\eps \in [0,1]$, we obtain the following sequence of PDEs:
\begin{align}
&O(\eps^0): &
0
	&= \( \d_t + \Lc_0\) h_0 + (f^{-1})_0 h_0^2 + (f^{-1} g)_0 h_0 + ( 4^{-1} f^{-1} g^2)_0 -\varphi, &
h_0(T,\cdot,\cdot)
	&=- \kappa, \label{pde:h0}\\
&O(\eps^n): &
0
	&= \( \d_t + \Lch_0 \) h_n + F_n , &
h_n(T,\cdot,\cdot)
	&= 0, \label{pde:hn}
\end{align}
where we have defined the differential operator
\begin{align}
\Lch_0
	:= \Lc_0 + 2 (f^{-1})_0 h_0 + (f^{-1}g)_0, \label{def:Lch}
\end{align}
and the functions
\begin{align}
F_n
	&:= \sum_{i=0}^{n-1} \Lc_{n-i} h_i + \sum_{i=0}^{n-1} \sum_{j=0}^i (f^{-1})_{n-i} h_{i-j} h_j + f_0 \sum_{i=1}^{n-1} h_{n-i}h_i
		+ \sum_{i=0}^{n-1} (f^{-1} g)_{n-i} h_i + (4^{-1} f^{-1} g^2)_n. \label{eq:Fn}
\end{align}
Equation \eqref{pde:hn} holds for all $n \geq 1$. We are now in position to define the $N$th order approximation for $h$.

\begin{definition}
Let $N$ be a non-negative integer, and assume that the coefficients of $\Lc$ and $\Nc$ are $C^N(D)$ where $D$ is an open set in $\Rb^2$. 
For any $(a,b) \in D$, we define the $N$th-order approximation of the transformed value function $h$ by
\begin{align}
\hb_N(t,a,b)
	&:= \sum_{n=0}^N \eps^n h_n(t,a,b) \Big|_{\(\eps, \ab,\bb\) = (1,a,b)}, \label{def:h.bar}
\end{align}
where $h_0$ is the solution to \eqref{pde:h0} and $h_n$, for $n \geq 1$, is the solution to \eqref{pde:hn}. 
\end{definition}

We now focus on developing an $N$th order approximation for the optimal liquidation strategy $\nu^\ast$.  To this end, recalling the expression \eqref{eq:nu} for the optimal execution strategy $\nu^\ast$ we define 
\begin{align}
\(\nu^\ast\)^\eps (\cdot,q,\cdot,\cdot)
	&:=  -\( \frac{ g^\eps + 2 h^\eps}{2 f^\eps}\)q, & 
\eps 
	&\in \[0,1\], \label{eq:nu.eps}
\end{align}
where $f^\eps$ and $g^\eps$ are given by \eqref{eq:chi.eps} and $h^\eps$ is the solution to \eqref{pde:h.eps}. 

\begin{definition}
Let $N$ be a non-negative integer, and assume that the coefficients of \eqref{eq:nu} are $C^N(D)$ where $D$ is an open set in $\Rb^2$. 
For any $(a,b) \in D$, we define the define the $N$th-order approximation of the optimal control $\nu^\ast$ as
\begin{align}
\nuba{N}(\cdot,\cdot,a,b)
	&:= \sum_{n=0}^N \eps^n \nu_n^\ast(\cdot,\cdot,a,b) \Big|_{\(\eps, \ab,\bb\) = (1,a,b)}, \label{def:nu.bar}
\end{align}
where, for every $n$, the function $\nu_n^\ast$ is the $n$th-order coefficient in the Taylor series expansion of $(\nu^\ast)^\eps$ about $\eps = 0$.
\end{definition}

\begin{remark}
As we noted at the beginning of Section \ref{sec:pde.asymptotics}, for a given $N$, the analyticity of the coefficients of $\Lc$ and $\Nc$ is not required to construct the approximation $\hb_N$. Indeed, to construct the $N$-th order approximations $\hb_N$ and $\nuba{N}$ one only needs that the coefficients are $C^N(D)$ for some $D \subseteq \Rb^2$.
\end{remark}

\begin{remark}
Observe that we have set $\eps =1$ in \eqref{def:h.bar} and \eqref{def:nu.bar}. As a result, this parameter plays no role in the approximations $\hb_N$ and $\nuba{N}$
(as it should not, as $\eps$ does not appear in the dynamics of $(X,S,Q,a,b)$ nor in the performance criteria $H^\nu$). Indeed, $\eps$ was introduced merely as an accounting tool in the formal asymptotic expansion performed above.
\end{remark}

\begin{remark}
Note that we have set $(\ab, \bb) = (a,b)$ in both \eqref{def:h.bar} and \eqref{def:nu.bar}. This is often a point of confusion, and we wish to make it clear how this is handled. First, we solve the sequence of PDE problems \eqref{pde:hn} with $(\ab,\bb)$ fixed. Let us make explicit the dependence of the solution of the $O(\eps^n)$ problem \eqref{pde:hn} on $(\ab,\bb)$ by writing $h_n^{(\ab,\bb)}(t,a,b)$. When we wish to compute the approximate value of $h$ at a point $(a,b)$, we evaluate $h_n^{(\ab,\bb)}(t,a,b)|_{(\ab,\bb) = (a,b)}$ for each $h_n$ appearing in the sum \eqref{def:h.bar}. Similarly, we can make explicit the dependence on $(\ab,\bb)$ of the $n$-th order component of $(\nu^\ast)^\eps$ by writing $(\nu_n^\ast)^{(\ab,\bb)}$. To compute the approximation for $\nu^\ast$ at a point $(a,b)$ we evaluate $(\nu_n^\ast)^{(\ab,\bb)}|_{(\ab,\bb) = (a,b)}$ for each term in the series \eqref{def:nu.bar}. The reason for choosing $(\ab,\bb) = (a,b)$ is as follows. The small-time behavior of a diffusion is predominantly determined by the geometry of the diffusion coefficients near the starting point of the diffusion $(a,b)$. In turn, the most accurate Taylor series expansion of any function near the point $(a,b)$ is the Taylor series centered at $(\ab,\bb) = (a,b)$. 
\end{remark}

We now give a representation of the approximate strategy $\nuba N$ \eqref{def:nu.bar} in terms of the functions $(h_n)$, which are solutions of the sequence of PDEs  \eqref{pde:hn}. 
\begin{proposition}
Fix $N \geq 0$, and suppose the coefficients functions appearing in \eqref{eq:chi} are $C^N(D)$ for some open set $D \subset \Rb^2$. Then for any $(a,b) \in D$, the approximate strategy $\nuba N$ in \eqref{def:nu.bar} is given by
\begin{align}
\nuba N(t,q,a,b)
    &= - \frac 1{f(a)} \( \frac 12 g(b) +  \sum_{n=0}^N h_n(t,a,b)\Big|_{(\ab,\bb) = (a,b)} \) q , 
    \label{eq:nu.bar}
\end{align}
where the $h_0$ is the solution to the PDE \eqref{pde:h0}, and $h_n$ is the solution to the PDE \eqref{pde:hn} for $n \geq 1$.
\end{proposition}
\begin{proof}
Suppose that $1 \leq k \leq N$. By \eqref{eq:chi.n}, we have that 
\begin{align}
\chi_k \Big|_{(\ab,\bb) = (a,b)}
	&= 0.
\end{align}
Thus, for $0 \leq n \leq N$,
\begin{align}
\nu_n^\ast(t,q,a,b) \Big|_{(\ab,\bb) = \(a,b\)}
	&= - \( \frac 12  (f^{-1} g)_n + \sum_{i=0}^n ( f^{-1})_{n-i} h_i(t,a,b) \)q \Bigg|_{(\ab,\bb) = \(a,b\)} \\
	&= - \frac 1{f(a)} \(\frac 12 g(b) \mathds{1}_{\{n=0\}} + h_n(t,a,b) \Big|_{\(\ab,\bb\) = (a,b)} \)q, \label{eq:nuast.eval} \\
\end{align}
where $\mathds{1}$ is the indicator function. By inserting \eqref{eq:nuast.eval} into \eqref{def:nu.bar} and evaluating at $\eps = 1$, we arrive at \eqref{eq:nu.bar}. 
\end{proof}

%
%

\subsection{Expressions for \texorpdfstring{$h_n$}{hn}}
\label{sec:pde.solution}

We begin this section by solving \eqref{pde:h0} explicitly for $h_0$, which yields explicit representation of the operator $\Lch_0$. We then give a recursive, integral expression for $h_n$ and evaluate the integral explicitly for $h_1$. We use the expressions for $h_0$ and $h_1$ to construct $\nuba 0$ and $\nuba 1$. For readability, we opt not to give $h_2$ or higher order approximations to the transformed value function. However, while tedious to obtain, their explicit computation is straightforward. 

\begin{proposition}
\label{prop:h0}
The solution $h_0$ to \eqref{pde:h0} is
\begin{align}
h_0(t)
	&=- \frac 12 g_0 + \sqrt{ \varphi f_0}\, \theta_0(t), &
\theta_0(t)
	&=\frac{ 1 + \zeta \ee^{ 2 \gam (T-t)}}{1 - \zeta \ee^{2\gam (T-t)}}, \label{eq:h0}
\end{align}
where we have defined the constants
\begin{align}
\gam 
	&:= \sqrt{ \frac{\varphi}{f_0}}, &
\zeta
	&:= \frac{ \kappa - \tfrac 12 g_0 + \sqrt{ \varphi f_0 }}{\kappa - \tfrac 12 g_0 - \sqrt{ \varphi f_0}}. \label{def:gam.zeta}
\end{align}

\end{proposition}

\begin{proof}
As both the forcing term and the terminal condition in \eqref{pde:h0} are independent of $(a,b)$, we conclude that $h_0$ is a function of $t$ only. Therefore, $\Lc_0 h_0 = 0$. The PDE \eqref{pde:h0} thus reduces to the constant coefficient ODE
\begin{align}
h_0' + \frac{1}{f_0} h_0^2 + \frac{g_0}{f_0} h_0 + \frac{g_0^2}{4 f_0} - \varphi
	&= 0, &
h_0(T)
	&= - \kappa. \label{ode:h0}
\end{align}
The reader will recognize \eqref{ode:h0} as a Ricatti equation. We check by direct substitution that \eqref{eq:h0} satisfies \eqref{ode:h0}. 
\end{proof}

\begin{corollary}
The operator $\Lch_0$, defined in \eqref{def:Lch}, is an elliptic operator and has the explicit representation 
\begin{align}
\Lch_0
	:=\Lc_0  + 2\gam \theta_0, \label{eq:Lh}
\end{align}
where $\Lc_0$ is defined in \eqref{eq:L}, $\theta_0$ is defined in \eqref{eq:h0}, and $\gam$ is given in \eqref{def:gam.zeta}.
\end{corollary}

With an explicit expression for $h_0$ in hand, we are able to write the zeroth order approximation $\nuba 0$ to the optimal liquidation strategy $\nu^\ast$. By \eqref{def:nu.bar} and \eqref{eq:h0}, we have 
\begin{align}
\nuba 0(t,q,a,b)
	&=- \gam \frac{ 1 + \zeta \ee^{ 2 \gam (T-t)}}{1 - \zeta \ee^{2\gam (T-t)}} q \Bigg|_{(\ab,\bb) = (a,b)},
	\label{eq:nu0}
\end{align}
where $\gam$ and $\zeta$ are given in \eqref{def:gam.zeta}. The strategy \eqref{eq:nu0} has the same form as the continuous time \cite{almgren2001optimal} strategy
, which we denote by $\nu_{AC}$. The difference between the strategies is that the level of the stochastic process $(a,b)$ is an input of zeroth order approximation $\nuba 0$ while in the Almgren-Chriss strategy $\nu_{AC}$ the price impact parameters $(a,b)$ are constants. Thus, the strategy $\nuba 0$ can be viewed as an implementation of $\nu_{AC}$ in which the price impact parameters are recalibrated in continuously time. 

Before we give an expression for $h_n$, let us review Duhamel's principle. Let $\Gamh_0$ be the fundamental solution of the operator $\d_t + \Lch_0$.  That is, $\Gamh_0$ satisfies the PDE
\begin{align}
0
    &=  \(\d_t + \Lch_0 \)\Gamh_0(\cdot,\cdot,\cdot;T,\alpha,\beta) , &
\Gamh_0(T,\cdot,\cdot;T,\alpha,\beta)
    &=  \del_{\alpha,\beta}, \label{pde:Gamh0}
\end{align}
where $\del_{\alpha,\beta}$ is the Dirac delta function on $\Rb^2$ centered at $(\alpha,\beta)$. By Duhamel's principle, the unique classical solution to a PDE of the  form
\begin{align}
0
    &=  (\d_t + \Lch_0)u + F , &
h(T,\cdot,\cdot)
    &= G,
\end{align}
is given by
\begin{align}
u(t,a,b)
    &=  \Pc_0(t,T)G(a,b) + \int_t^T \dd s \, \Pc_0(t,s)F(s,a,b), 
\end{align}
where we have introduced the semigroup $\Pc_0$ generated by $\Lch_0$, which is defined as
\begin{align}
\Pc_0(t,s) G(a,b)
    &=  \int_{\Rb^2} \dd \alpha \dd \beta \, \Gamh_0(t,a,b;s,\alpha,\beta) G(\alpha,\beta), \label{eq:P0}
\end{align}
for $0 \leq t \leq s \leq T$.

We now give a recursive expression for $h_n$. 
\begin{proposition}
\label{prop:hn}
The solution $h_n$ to \eqref{pde:hn} is
\begin{align}
h_n(t,a,b)
	&= \int_t^T \dd s\, \Pc_0(t,s) F_n(s,a,b)
	= \int_t^T \dd s \int_{\Rb^2} \dd \alpha \dd \beta\, \Gamh_0(t,a,b;s,\alpha,\beta) F_n(s,\alpha,\beta). \label{eq:hn}
\end{align}
Here, $\Pc_0$ is given in \eqref{eq:P0}, $F_n$ is given in \eqref{eq:Fn}, and $\Gamh_0$ is given by 
\begin{align}
\Gamh_0(t,\av;s,\yv)
	&= \frac{\Psi_0(t,s)}{ 2 \pi \sqrt{ \det \Cv(t,s)}}
	\exp\(	-\frac 12 \< \Cv^{-1}(t,s)(\yv - \av - \mv(t,s)), \yv - \av - \mv(t,s) \> \), \label{eq:Gamh}
\end{align}
where 
\begin{subequations}
\label{eqs:Gam0consts}
\begin{align}
C
	&:= \begin{pmatrix} \om_0^2 & \rho ( \om \psi)_0 \\ \rho (\om\psi)_0 & \psi_0^2 \end{pmatrix}, &
\Cv(t,s)
	&:= \int_t^s \dd s\, C
	= C(s-t), \label{eqs:C} \\
\av
	&:= \begin{pmatrix} a & b \end{pmatrix}^\top, &
\yv
	&:= \begin{pmatrix} \alpha & \beta \end{pmatrix}^\top, \label{eqs:ay} \\
m
	&:= \begin{pmatrix} \mu_0 & \eta_0 \end{pmatrix}^\top, &
\mv(t,s)
	&:= \int_t^s \dd s\, m
	= (s-t) m, \label{eqs:m}
\end{align}
\end{subequations}
and 
\begin{align}
\Psi_0(t,s)
	&:= e^{-2 \gam  (s-t)}\left(\frac{ \zeta  e^{2 \gam  T}-e^{2 \gam  s}}{\zeta  e^{2 \gam  T}-e^{2 \gam  t}}\right)^2. \label{eq:Psi}
\end{align}
\noeqref{eqs:C}\noeqref{eqs:ay}\noeqref{eqs:m}
\end{proposition}
\begin{proof}
Let $\Gamh_0$ be the solution to the PDE \eqref{pde:Gamh0}. As the coefficients of the operator $\Lch_0$ are constant in $(a,b)$ we have
\begin{align}
\Gamh_0(t,\av;s,\yv)
	&= \frac{1}
	{ 2 \pi \sqrt{ \det \Cv(t,s)}} \exp\(2\gam \int_t^s \dd r\, \theta_0(r)
	-\frac 12 \< \Cv^{-1}(t,s)(\yv - \av - \mv(t,s)), \yv - \av - \mv(t,s) \> \), \label{eq:Gamh.pre}
\end{align}
where $\Cv, \av, \yv$ and $\mv$ are given in \eqref{eqs:Gam0consts}. Using \eqref{eq:h0}, we compute explicitly
\begin{align}
\exp\(2\gam \int_t^s \dd r\, \theta_0(r)\)
	&= \Psi_0(t,s),
\end{align}
which yields the expression \eqref{eq:Gamh}. Applying Duhamel's principle to \eqref{pde:hn} gives \eqref{eq:hn}.
\end{proof}

\begin{corollary}\label{cor:h1}
Define 
\begin{subequations}
\label{def:h1.summands}
\begin{align}
c^{(1)}\(t,a,b\)
	&:= -\gam^2 f'(\ab) \mu_0, &
I^{(1)}(t)
	&:=\int_t^T \dd s\, s \theta_0^2(s) \Psi_0(t,s),\label{def:h1.part1}\\
c^{(2)}\(t,a,b\)
	&:= -\gam^2 f'(\ab) (a - \ab - t \mu_0), &
I^{(2)}(t)	
	&:= \int_t^T\dd s\, \theta_0^2(s) \Psi_0(t,s),\label{def:h1.part2}\\
c^{(3)}\(t,a,b\)
	&:= \gam  g'(\bb)\eta_0, &
I^{(3)}(t)
	&:= \int_t^T \dd s\, s \theta_0(s) \Psi_0(t,s),  \label{def:h1.part3} \\ 
c^{(4)}\(t,a,b\)
	&:= \gam g'(\bb)(b - \bb - t \eta_0), &
I^{(4)}(t)
	&:= \int_t^T \dd s\, \theta_0(s) \Psi_0(t,s), \label{def:h1.part4}  \noeqref{def:h1.part1}\noeqref{def:h1.part2}\noeqref{def:h1.part3}\noeqref{def:h1.part4}
\end{align}
\end{subequations}
where $\theta_0$ and $\Psi_0$ are given in \eqref{eq:h0} and \eqref{eq:Psi}, respectively. Then the solution $h_1$ to \eqref{pde:hn} is given by 
\begin{align}
h_1(t,a,b)
	&= \sum_{i=1}^4 c^{(i)}\(t,a,b\) I^{(i)}\(t\). \label{eq:h1}
\end{align}
We evaluate the integrals $I^{(i)}$ explicitly in \eqref{eqs:I}.
\end{corollary}
\begin{proof}
See Appendix \ref{app:h1proof}.
\end{proof}
With an explicit expression for $h_1$ in hand, we are able to construct the first order approximation $\nuba 1$ to the optimal liquidation strategy $\nu^\ast$. By \eqref{eq:nu.bar} and \eqref{eq:nu0}, we have 
\begin{align}
\nuba 1 (t,q,a,b)
	&= - \(  \gam \frac{ 1 + \zeta \ee^{ 2 \gam (T-t)}}{1 - \zeta \ee^{2\gam (T-t)}}  
		+ \frac 1{f(a)} \sum_{i=1}^4 c^{(i)}\(t,a,b\) I^{(i)} \(t\)\)q \Bigg|_{(\ab,\bb) = (a,b)}. \label{eq:nu1}
\end{align}
Higher order approximations to the transformed value function $h$ can be computed explicitly. For the sake of readability, we do not carry out these calculations and instead focus on zero and first order approximations, which are sufficient to capture the lowest order effects of stochastic price impact.

%
%

\subsection{Analysis of some limiting cases}
\label{sec:limstrats}

The asymptotic approximations of the transformed value function $h$ (and hence the approximations to $\nu^\ast$) developed in Section \ref{sec:asymptotics} depend on the parameters $\kappa$, which controls the penalty for liquidation occurring at the trading horizon $T$, and $\varphi$, which controls the penalty for holding shares of $S^\nu$ throughout trading. In this section, we develop strategies that are independent of one or both of these parameters by taking the limit of the optimal strategy approximations $\nuba N$ as the parameters $\kappa$ and $\varphi$ tend to $\infty$ and $0$, respectively. One motivation for considering such strategies is financial. Taking the limit of the optimal strategy $\nu^\ast$ as $\kappa \to \infty$ corresponds to a setting in which the trader is adamant about his entire inventory being liquidated before $T$. Then, taking $\varphi \to 0$ corresponds to a setting in which the trader both demands complete liquidation by the trading horizon and is indifferent about holding inventory. Another motivation for developing limiting strategies is analytic tractability. We shall see later in this section that the limiting strategy approximations $\kappa \to \infty$ and $(\kappa,\varphi) \to (\infty,0)$ have far fewer terms than their corresponding nonlimiting case strategies, thereby facilitating a financial interpretation of the resulting expressions. 

For the remainder of this section, we make the dependence on parameters $\kappa$ and $\varphi$ of the PDE solutions $h_n$ and strategies $\nu$ explicit with the superscipt notation
\begin{align}
h_n 
    &\equiv h_n^{(\kappa,\varphi)}, &
\nu
    &\equiv \nu^{(\kappa, \varphi)}.
\end{align}
We refer to the strategies $\nu^{(\kappa,\varphi)}$ with $k \neq \infty$ and $(\kappa, \varphi) \neq (\infty,0)$ as \textit{nonlimiting strategies}. Let us define the \textit{limiting strategies}
\begin{align}
\nu_{AC}^{(\infty, \varphi )}
    &:= \lim_{\kappa \to \infty} \nu_{AC}^{(\kappa, \varphi)}, &
\nu_{AC}^{(\infty, 0)}
    &:= \lim_{\varphi \to 0} \lim_{\kappa \to \infty} \nu_{AC}^{(\kappa, \varphi)}, \\ 
\(\nu^\ast\)^{(\infty, \varphi )}
    &:= \lim_{\kappa \to \infty} \(\nu^\ast\)^{(\kappa, \varphi)}, &
\(\nu^\ast \)^{(\infty, 0)}
    &:= \lim_{\varphi \to 0} \lim_{\kappa \to \infty} \(\nu^\ast \)^{(\kappa, \varphi)}, \\   
\nuba N^{(\infty,\varphi)}
    &:= \lim_{\kappa \to \infty} \nuba N^{(\kappa, \varphi)}, &
\nuba N^{(\infty,0)}
    &:= \lim_{\varphi \to 0} \lim_{\kappa \to \infty} \nuba N^{(\kappa, \varphi)}.
\end{align}
Below, we provide explicit expressions for $\nuba N^{(\infty,\varphi)}$ and $\nuba N^{(\infty,0)}$ for $N \in \{0,1\}$. In order to construct $\nuba N^{(\infty,\varphi)}$ and $\nuba N^{(\infty,0)}$, it will be helpful to define
\begin{align}
h_n^{(\infty, \varphi)}(t)
	&:= \lim_{\kappa \to \infty}  h_n^{(\kappa, \varphi)}(t), &
h_n^{(\infty, 0)}(t)
	&:= \lim_{\varphi \to 0} \lim_{\kappa \to \infty} h_n^{(\kappa, \varphi)}(t) . \label{eq:hn.limiting}
\end{align}
As the strategies $\nuba N^{(\kappa, \varphi)}$ depend on $(\kappa, \varphi)$ only through the functions $h_n^{(\kappa,\varphi)}$, we can obtain the strategies $\nuba N^{(\infty,\varphi)}$ and $\nuba N^{(\infty,0)}$ by replacing $h_n^{(\kappa,\varphi)}$ in \eqref{eq:nu.bar} with $h_n^{(\infty,\varphi)}$ and $h_n^{(\infty,0)}$, respectively. From \eqref{eq:h0}, a straightforward computation yields 
\begin{align}
h_0^{(\infty,\varphi)}(t)
    &= - \frac{g_0}{2} - f_0 \gam \coth\( \gam (T-t)\), &
h_0^{(\infty,0)}(t)
    &= -\frac{g_0}{2}- \frac{f_0}{T-t}. \label{eq:h0.limiting}
\end{align}
We compute $h_1^{(\infty,\varphi)}$ and $h_1^{(\infty,0)}$ in Appendix \ref{app:h1} to obtain
\begin{align}
h_1^{(\infty,\varphi)}(t,a,b)
	&= \sum_{i=1}^4 c^{(i)}(t,a,b) I_{(\infty,\varphi)}^{(i)}(t), \label{eq:hL} \\
h_1^{(\infty,0)}(t,a,b)
    &= -\frac{ f'(\ab)}{2(T-t)} \(2(a - \ab) + \mu_0(T-t)\)
        - \frac {g'(\bb)}6 \(3(b - \bb) + \eta_0 ( T-t)\), \label{eq:hLH}
\end{align}
where $c^{(i)}$ and $I_{(\infty,\varphi)}^{(i)}$ are given in \eqref{def:h1.summands} and \eqref{eqs:I.limits.1}, respectively. For $n \in \{ 0, 1 \}$, replacing $h_n^{(\kappa,\varphi)}$ with $h_n^{(\infty,\varphi)}$ in \eqref{eq:nu.bar} yields the strategies
\begin{align}
\nuba 0^{(\infty,\varphi)} (t,q,a,b)
	&=  \gam \coth\( \gam (T-t) \) q \Bigg|_{(\ab,\bb) = (a,b)}, \\
\nuba{1}^{(\infty,\varphi)} (t,q,a,b)
	&= \(  \gam \coth\( \gam (T-t) \) - \frac 1{f(a)} \sum_{i=1}^4 c^{(i)} (t,a,b) I_{(\infty,\varphi)}^{(i)}(t) \)q \Bigg|_{(\ab,\bb) = (a,b)},
\end{align}
and replacing $h_n^{(\kappa,\varphi)}$ in \eqref{eq:nu.bar} with $h_n^{(\infty,0)}$ in \eqref{eq:nu.bar} yields the strategies
\begin{align}
\nuba 0^{(\infty,0)}(t,q,a,b)
	&= \frac{1}{T-t} q, \label{eq:twa}\\
\nuba 1^{(\infty,0)}(t,q,a,b)
	&= \( \frac{1}{T-t} + \frac 12  \frac{\mu(a) f'(a)}{f(a)} + \frac 16 (T-t) \frac{\eta(b)  g'(b) }{ f(a) }\) q. \label{eq:corrected.twa}
\end{align}

The reader will recognize \eqref{eq:twa} as the \textit{time-weighted average} strategy. We can thus view equation \eqref{eq:corrected.twa} as a first order correction to the time-weighted average strategy. The second term in \eqref{eq:corrected.twa} instructs the trader to adjust his trading speed in proportion to the product of the slope $f'$ of the temporary impact function and the drift $\mu$ of the process $a$. For instance, suppose that at time $t$, $\mu(a_t) < 0$ and $f'(a_t) > 0$. In that case, the temporary price impact process $a$ is drifting downwards, and the price impact $f(a)$ will decrease with $a$. The second term in \eqref{eq:corrected.twa} instructs the trader to slow down liquidation because he expects a lower temporary impact in the near future. If $\mu(a_t) > 0$ and $f'(a_t) > 0$, then a trader following $\nuba 1^{(\infty,0)}$ will speed up trading as he expects higher price impact soon. Of course, if $\mu(a_t) f'(a_t)$ is small relative to $f(a)$, then the contributions from this term are small. The third term instructs the trader to adjust his trading speed proportional to the product $\eta(b) g'(b)$. Like the second term, the adjustments of the third term are weighted relative to the temporary price impact. But unlike the second term, the third term's influence on the trading speed diminishes as time approaches the trading horizon. This is intuitive as permanent impact matters not to a trader who is soon to exit the market. We note that the third term in strategy \eqref{eq:corrected.twa} causes more dramatic deviations from the time-weighted strategy when the permanent price impact is large relative to the temporary price impact early in the trading period.

%
%

\section{Numerical examples}
\label{sec:simulation}

In this section, we provide examples of simulated trading using the strategies we developed in Sections \ref{sec:pde.solution} and \ref{sec:limstrats}. In Example \ref{ex:first-order-buy}, we simulate a single trading day to illustrate the effects of the correction terms present in the first order limiting strategy $\nuba 1^{(\infty,0)}$. In Example \ref{ex:monte-carlo}, we simulate a large number of trading days to demonstrate the improvement over the Almgren-Chriss strategy our approximations give in the nonlimiting and limiting cases. Furthermore, we demonstrate an improvement of the first order strategies over the zeroth order strategies in the nonlimiting and limiting cases. 

Throughout this section, we assume the price impact processes $a$ and $b$ are Cox-Ingersoll-Ross (\cite{cox1985theory}) processes with the dynamics
\begin{align}
\mu(z)
	&= \lambda_a \( \theta_a - z \) , & 
\omega(z) 
	&= \sig_a \sqrt z, \\
\eta(z)
	&= \lambda_b\( \theta_b -z \), & 
\psi(z) 
	&= \sig_b \sqrt z.
\end{align}
where the constants $\lambda_a, \theta_b, \theta_a, \theta_b, \sig_a$ and $\sig_b$ are all positive, and the Brownian motions $B_t^{(1)}$ and $B_t^{(2)}$ are correlated with parameter $\rho$. Furthermore, we require that the coefficients $\lambda_a, \theta_b, \theta_a, \theta_b, \sig_a$ and $\sig_b$ satisfy the Feller condition
\begin{align}
2 \lambda_a \theta_a 
    &> \sig_a^2, &
2 \lambda_b \theta_b 
    &> \sig_b^2, \label{eq:feller}
\end{align}
so that $a$ and $b$ are strictly positive processes. Explicitly, we have
\begin{align}
\dd a_t 
    &= \lambda_a \( \theta_a - a_t \) \dd t + \sig_a \sqrt{a_t} \dd B_t^{(2)}, &
\dd b_t 
    &= \lambda_b\( \theta_b - b_t \) \dd t + \sig_b \sqrt{b_t} \dd B_t^{(1)}.
\label{eq:impact.dynamics.ex1}
\end{align}
We also take 
\begin{align}
f(a) &= a , &
g(b) &= b.
\end{align} 
Note that because we require the Feller condition \eqref{eq:feller} to be satisfied we have $f(a_t) > 0$ and $g(b_t) > 0$ for all $0 \leq t \leq T$.  Thus, both temporary and permanent price impact processes $f(a)$ and $g(b)$ remain strictly positive.

\begin{example}
\label{ex:first-order-buy}
Let us suppose the trader demands complete liquidation by $T$ and is indifferent to holding inventory. In this case, the zeroth order strategy $\nuba 0^{(\infty,0)}$ is equal to time-weighted average strategy \eqref{eq:twa} and thus does not depend on $(a,b)$. Under the dynamics \eqref{eq:impact.dynamics.ex1}, the first order strategy $\nuba 1^{(\infty,0)}$ is given by 
\begin{align}
\nuba 1^{(\infty,0)} (t,q,a,b)
    &= \( \frac 1{T-t} + \frac{\lambda_a \(\theta_a - a\)}{2 a} + (T-t)\frac{\lambda_b\(\theta_b  - b \)}{6a}   \)q. \label{eq:ex-strat-order1-LH}
\end{align}

The second and third terms in \eqref{eq:ex-strat-order1-LH} are a time-dependent linear combination of the distance of the price impact processes from their respective long-run means relative to $a$. The mean reversion parameters $\lambda_a$ and $\lambda_b$ control the aggressiveness of the adjustment. When the mean reversion parameters are large, we expect that deviations of $a$ and $b$ from their respective means $\theta_a$ and $\theta_b$ to be short lived. In this case, the strategy $\nuba 1^{(\infty,0)}$ adjusts quickly to take advantage of these deviations.

When a trader is following the time-weighted average strategy, at each instant he sells a fraction of his inventory that is inversely proportional to the remaining time $T-t$. As such $\nuba 0^{(\infty,0)} > 0$ for all $t$ in the trading period. In some instances, however, the strategy $\nuba 1^{(\infty,0)}$ instructs the trader to purchase shares of $S^\nu$ This occurs in times of relatively large price impact. Although it may seem counter-intuitive for a trader who wishes to liquidate a position to buy shares, this strategy can increase the objective function $H^\nu$ if, for example, the trader buys shares of $S^\nu$ during a period of relatively high price impact, putting upward pressure on the midprice $S^\nu$, then subsequently sells shares rapidly during a period of low price impact. 

In Figure \ref{fig:ex1.2}, we provide a simulated path of $(a,b)$ and the paths $(X^\nu,S^\nu, Q^\nu)$ that result from following strategies $\nu = \nuba 0^{(\infty,0)}$ and $\nu = \nuba 1^{(\infty,0)}$ using parameters
\begin{align}
\left. \begin{aligned}
\lambda_a 
    &= 10, & 
\theta_a 
    &= 2 \times 10^{-6}, & 
\sig_a 
    &= 1.5 \times 10^{-3},  \\
\lambda_b 
    &= 10, &
\theta_b 
    &= 5 \times 10^{-5}, & 
\sig_b 
    &= 3 \times 10^{-3}, \\
T 
    &= 1, &
\rho 
    &= 0.7, &
\sig
    &= 0.01.
\end{aligned} \right\} \label{eq:ex1.param.values}
\end{align}
This simulation demonstrates how the first order strategy $\nuba 1^{(\infty,0)}$ responds to the high and low price impact values encountered early in the trading period. The strategy $\nuba 1^{(\infty,0)}$ instructs the trader to purchase shares of $S^\nu$ when price impact is relatively high early in the trading period. The price impact processes subsequently decrease below their long run means, and the trader following $\nuba 1^{(\infty,0)}$ liquidates shares at rate in excess of the trading speed dictated by $\nuba 0^{(\infty,0)}$. In this example, $\nuba 1^{(\infty,0)}$ is more profitable than $\nuba 0^{(\infty,0)}$.

While the the optimal strategy approximations we have developed in this paper suggest that, under certain market conditions, the trader should buy shares of the stock when the price impact is high, the authors of \cite{bertsimas1999optimal} note that in practice, if a trader wants to sell a block of securities then it is usually antithetical to their stance as a seller to purchase shares of the security during the trading period. In some cases, it is a violation of a manager's fiduciary responsibility to their client and is hence illegal. As such, a liquidation strategy $\nu$ that can sometimes instruct a trader to buy could be modified to be $\max(0,\nu)$. A truncated strategy may not be optimal with respect to the objective functional $H^\nu$ that we have defined,  but in our numerical simulations the periods in which our approximate strategies instruct the trader to buy are short-lived.
\end{example}

\begin{example} \label{ex:monte-carlo}

In this example, we carry out a number of Monte Carlo simulations to evaluate the performance of the liquidation strategies $\nuba N$ as well as the limiting strategies $\nuba N^{(\infty,\varphi)}$  and $\nuba N^{(\infty,0)}$ for $N \in \{ 0,1\}$. We demonstrate the relative improvement a trader gains by following $\nuba 0$ over $\nu_{AC}$ and the relative improvement a trader gains by following $\nuba 1$ over $\nuba 0$. We repeat this experiment in the limiting case $\kappa \to \infty$. We see from \eqref{eq:twa} that in the limiting case $(\kappa, \varphi) \to (\infty,0)$, both $\nu_{AC}^{(\infty,0)}$ and $\nuba 0^{(\infty,0)}$ are equal to the time-weighted average strategy (i.e. $\nu_{AC}^{(\infty,0)} = \nuba 0^{(\infty,0)} = q (T-t)^{-1}$). So, we demonstrate the relative performance increase a trader gains by following $\nuba 1^{(\infty,0)}$ over $\nuba 0^{(\infty,0)}$. 

To this end, let us define
\begin{align}
\Phi(\nu)
    &:= X_T^\nu + Q_T^\nu (S_T^\nu - \kappa Q_T^\nu) - \varphi \int_0^T \dd s\, \( Q_s^\nu\)^2, &
\Phi^{(\infty,\varphi)}(\nu)
    &:= X_T^\nu - \varphi \int_0^T \dd s\, \( Q_s^\nu\)^2, \label{eq:pref.criteria} \\
\Phi^{(\infty,0)}(\nu) 
    &:= X_T^\nu.
\end{align}
For a fixed strategy $\nu$, the random variables $\Phi(\nu)$, $\Phi^{(\infty,\varphi)}(\nu)$ and $\Phi^{(\infty,0)}(\nu)$ give the value a trader following $\nu$ achieves on a single path of $(X_T^\nu,S_T^\nu,Q_T^\nu,a,b)$. 
In both of the limiting cases $\kappa \to \infty$ and $(\kappa,\varphi) \to (\infty,0)$, the optimal strategies ensure liquidation by the terminal time $T$. Therefore, the term $Q_T^\nu (S_T^\nu - \kappa Q_T^\nu)$ that accounts for liquidation of the remaining shares at time $T$ does not appear in either $\Phi^{(\infty,\varphi)}(\nu)$ or $\Phi^{(\infty,0)}(\nu)$. Let us also define the sample mean of our Monte carlo simulations as follows
\begin{align}
\widehat \Phi (\nu) 
    &:= \frac 1M \sum_{i=1}^M \Phi_i(\nu), \label{eq:sample.mean}
\end{align}
where $\Phi_i(\nu)$ is the value of $\Phi(\nu)$ obtained by the $i$-th independent path of $(X_T^\nu,S_T^\nu,Q_T^\nu,a,b)$.  Observe that $\widehat{\Phi}(\nu)$ is a statistical estimate of the performance criteria $H^\nu$. The definitions for $\widehat\Phi^{(\infty,\varphi)}(\nu)$ and $\widehat\Phi^{(\infty,0)}(\nu)$, our statistical estimators for $H^\nu$ in the limiting cases $\kappa \to \infty$ and $(\kappa,\varphi)\to(\infty,0)$, are analogous. 

In this example, we take the following parameters 
\begin{align}
\left. \begin{aligned}
\lambda_a 
    &= 1, & 
\theta_a 
    &= 1 \times 10^{-4}, & 
\sig_a 
    &= 8 \times 10^{-3},  \\
\lambda_b 
    &= 1, &
\theta_b 
    &= 5 \times 10^{-4}, & 
\sig_b 
    &= 8 \times 10^{-3}, \\
T 
    &= 1, &
\rho 
    &= 0.7, &
\varphi
    &= 0.01, \\
\kappa  
    &= 10, &
\sig
    &= 0.2.&
\end{aligned} \right\} \label{eq:ex2.param.values}
\end{align}
We note that the processes $a$ and $b$ under the parameter choice \eqref{eq:ex2.param.values} both satisfy the Feller condition \eqref{eq:feller}. Furthermore, we choose the initial conditions
\begin{align}
t 
    &= 0, & 
X_0 
    &= 0, & 
S_0 
    &= 40, & 
Q_0 
    &= 5000, & 
a_0
    &= \theta_a, &
b_0
    &= \theta_b. \label{eq:init.cond}
\end{align}
The price impact parameters in the Almgren-Chriss strategy $\nu_{AC}$ are constant, and we take them to be $(a_0,b_0) = (\theta_a, \theta_b)$. In total, we run $M = 10,000$ sample paths.

In our Monte Carlo simulations, we obtain
\begin{align}
\frac{\widehat\Phi\(\nuba 0\) - \widehat\Phi\(\nu_{AC}\)}{\widehat\Phi\(\nu_{AC}\)}\cdot 10^4    
    &= 6.0385, &
\frac{\widehat\Phi\(\nuba 1\) - \widehat\Phi\(\nuba{0}\)}{\widehat\Phi\(\nuba{0}\)}  \cdot 10^4
    &= 0.0224, \label{eq:centered-strat-results} 
\end{align}
in the nonlimiting case, 
\begin{align}
\frac{\widehat\Phi^{(\infty,\varphi)}\(\nuba 0^{(\infty,\varphi)}\) - \widehat\Phi^{(\infty,\varphi)}\(\nu_{AC}^{(\infty,\varphi)}\)}{\widehat\Phi^{(\infty,\varphi)}\(\nu_{AC}^{(\infty,\varphi)}\)}\cdot 10^4    
    &= 6.0367, \label{eq:centered-strat-results-L-1}\\
\frac{\widehat\Phi^{(\infty,\varphi)}\(\nuba 1^{(\infty,\varphi)}\) - \widehat\Phi^{(\infty,\varphi)}\(\nuba{0}^{(\infty,\varphi)}\)}{\widehat\Phi^{(\infty,\varphi)}\(\nuba{0}^{(\infty,\varphi)}\)}  \cdot 10^4
    &= 0.0224, \label{eq:centered-strat-results-L-2}
\end{align}
in the limiting case $\kappa \to \infty$, and 
\begin{align}
\frac{\widehat\Phi^{(\infty,0)}\(\nuba 1^{(\infty,0)}\) - \widehat\Phi^{(\infty,0)}\(\nuba{0}^{(\infty,0)}\)}{\widehat\Phi^{(\infty,0)}\(\nuba{0}^{(\infty,0)}\)}  \cdot 10^4
    &= 0.8131. \label{eq:centered-strat-results-LH}
\end{align}
in the limiting case $(\kappa,\varphi) \to (\infty,0)$. 

Equations \eqref{eq:centered-strat-results}, \eqref{eq:centered-strat-results-L-1}, \eqref{eq:centered-strat-results-L-2} and \eqref{eq:centered-strat-results-LH} demonstrate that in the nonlimiting case and both limiting cases, the trader gains a relative value increase from following the zeroth order strategy approximation over the Almgren-Chriss strategy and from following the first order strategy approximation over the zeroth order strategy approximation.

In Figures \ref{fig:trim-centered-zero-over-ac} and \ref{fig:trim-centered-first-over-zero}, we plot histograms of the relative performance
\begin{align}
&\frac{\Phi(\nuba 0) - \Phi(\nu_{AC})}{\Phi(\nu_{AC})}\cdot 10^4, &
&\frac{\Phi(\nuba 1) - \Phi(\nuba{0})}{\Phi(\nuba{0})}  \cdot 10^4, \label{eq:rel.perf}
\end{align}
respectively, in Figures \ref{fig:trim-centered-zero-over-ac-L} and \ref{fig:trim-centered-first-over-zero-L} we plot histograms of the relative performance
\begin{subequations}
\label{eqs:rel.perf.L}
\begin{align}
&\frac{\Phi^{(\infty,\varphi)}\(\nuba 0^{(\infty,\varphi)}\) - \Phi^{(\infty,\varphi)}\(\nu_{AC}^{(\infty,\varphi)}\)}{\Phi^{(\infty,\varphi)}\(\nu_{AC}^{(\infty,\varphi)}\)}\cdot 10^4, \label{eq:rel.perf.L.0} \\
&\frac{\Phi^{(\infty,\varphi)}\(\nuba 1^{(\infty,\varphi)}\) - \Phi^{(\infty,\varphi)}\(\nuba{0}^{(\infty,\varphi)}\)}{\Phi^{(\infty,\varphi)}\(\nuba{0}^{(\infty,\varphi)}\)}  \cdot 10^4, \label{eq:rel.perf.L.1}
\end{align}
\end{subequations}
respectively, and in Figure \ref{fig:trim-centered-first-over-zero-LP}, we plot a histogram of the relative performance
\begin{align}
\frac{\Phi^{(\infty,0)}\(\nuba 1^{(\infty,0)}\) - \Phi^{(\infty,0)}\(\nuba{0}^{(\infty,0)}\)}{\Phi^{(\infty,0)}\(\nuba{0}^{(\infty,0)}\)}  \cdot 10^4. \label{eq:rel.perf.LH}
\end{align}
Figures \ref{fig:trim-centered-zero-over-ac} and \ref{fig:trim-centered-first-over-zero} show that in addition to the expected value increases seen in \eqref{eq:centered-strat-results}, $\Phi_i (\nuba 0) > \Phi_i (\nu_{AC})$ and $\Phi_i(\nuba 1) > \Phi_i( \nuba 0)$ more often than not. We see the same result in both the limiting cases $\kappa \to \infty$ and $(\kappa, \varphi) \to (\infty,0)$. 

For the chosen parameters \eqref{eq:ex2.param.values} and initial conditions \eqref{eq:init.cond}, the relative improvement gained by following a first order strategy approximation is muted compared to the relative improvement of the zeroth order strategy over Almgren-Chriss. When $(a_0, b_0) = (\theta_a, \theta_b)$, the price impact parameters typically hover around their respective long-run means, keeping the correction terms in $\nuba 1$, $\nuba 1^{(\infty,\varphi)}$ and $\nuba 1^{(\infty,0)}$ small. Let us keep the parameter values \eqref{eq:ex2.param.values} but modify the initial conditions as follows
\begin{align}
t 
    &= 0, & 
X_0 
    &= 0, & 
S_0 
    &= 40, & 
Q_0 
    &= 5000, & 
a_0
    &= 1.5 \theta_a, &
b_0
    &= 1.5 \theta_b. \label{eq:init.cond.above}
\end{align}
We note that the difference between the initial conditions \eqref{eq:init.cond} and \eqref{eq:init.cond.above} are the values of $a_0$ and $b_0$. With the initial conditions \eqref{eq:init.cond.above}, the price impact processes $a$ and $b$ start above their long-run means $\theta_a$ and $\theta_b$ and will typically float downwards towards their respective means throughout the trading period.  When the price impact parameters begin away from their long-run means, the correction terms present in the first order strategy approximations have a more pronounced influence on the trading strategy, and we see a larger relative improvement of the first order strategies over the zeroth order strategies. 

We repeat the above experiments with the initial conditions \eqref{eq:init.cond.above} and obtain 
\begin{align}
\frac{\widehat\Phi\(\nuba 1\) - \widehat\Phi\(\nuba{0}\)}{\widehat\Phi\(\nuba{0}\)}  \cdot 10^4
    &= 0.2682, &
\frac{\widehat\Phi^{(\infty,\varphi)}\(\nuba 1^{(\infty,\varphi)}\) - \widehat\Phi^{(\infty,\varphi)}\(\nuba{0}^{(\infty,\varphi)}\)}{\widehat\Phi^{(\infty,\varphi)}\(\nuba{0}^{(\infty,\varphi)}\)}  \cdot 10^4
    &= 0.2683, \\& &
\frac{\widehat\Phi^{(\infty,0)}\(\nuba 1^{(\infty,0)}\) - \widehat\Phi^{(\infty,0)}\(\nuba{0}^{(\infty,0)}\)}{\widehat\Phi^{(\infty,0)}\(\nuba{0}^{(\infty,0)}\)}  \cdot 10^4
    &= 3.541.
\end{align}
The relative performance of the first order strategies increases by an order of magnitude when going from the initial conditions \eqref{eq:init.cond} to \eqref{eq:init.cond.above}. In Figures \ref{fig:trim-above-first-over-zero}, \ref{fig:trim-above-first-over-zero-L} and \ref{fig:trim-above-first-over-zero-LH}, we plot histograms of the relative performance \eqref{eq:rel.perf} (right), \eqref{eqs:rel.perf.L} and \eqref{eq:rel.perf.LH}, respectively, with the initial conditions \eqref{eq:init.cond.above}. We see that in all three cases, the first order strategy out-performs the zeroth order strategy more often than not.

\end{example}

%
%

\section{Conclusion}
\label{sec:conclusion}

In this paper, we present a formal approximation to the optimal trading strategy for a trader facing the liquidation problem under a market model in which the price impact factors are stochastic. Our model supposes general diffusion dynamics of the price impact factors $a$ and $b$ and allows for the price impact processes $f(a)$ and $g(b)$ to be nonlinear functions of the diffusions $a$ and $b$. The continuous-time Almgren-Chriss strategy is encapsulated in our model, and the zeroth order approximation to the optimal liquidation strategy is interpreted as the Almgren-Chriss strategy where the price impact parameters are continuously recalibrated. Higher-order strategy approximations take into account the geometry of the price impact processes diffusion coefficients, allowing a trader to adjust his trading strategy for times of relatively high and low price impact. We also demonstrate numerically that higher-order strategy approximations outperformed lower order approximations. 

%
%

\bibliographystyle{chicago}
\bibliography{bibfile}

%
%
\clearpage
\appendix

\section{Proof of Corollary \ref{cor:h1}}\label{app:h1proof}

\begin{proof}

By \eqref{eq:hn}, we have
\begin{align}
&h_1(t,a,b) 
    = \int_t^T \dd s\, \Pc_0(t,s) F_1(s,a,b) \\
    =& \int_t^T \dd s \int_{\Rb^2} \dd \alpha \dd \beta\, \Gamh_0(t,a,b;s,\alpha,\beta) \( (f^{-1})_1(\alpha) h_0^2(s) + (f^{-1}g)_1(\alpha,\beta) h_0(s) + (4^{-1} f^{-1} g^2)_1(\alpha,\beta) \). \label{eq:h1.pre}
\end{align}
Using \eqref{eq:chi.n}, we see that
\begin{align}
\( f^{-1}\)_1 ( \alpha)
    &= - \frac{f'(\ab)}{f^2(\ab)} ( \alpha - \ab), &
(f^{-1}g)_1(\alpha,\beta)
    &= \frac{f'(\ab) g(\bb)}{f^2(\ab)} (\alpha - \ab) + \frac{g'(\bb)}{f(\ab)} (  \beta- \bb), \label{eq:fun1.line1}
\end{align}
and 
\begin{align}
\(4^{-1} f^{-1} g^2\)_1(\alpha,\beta)
    &= -\frac{f'(\ab) g^2(\bb)}{4 f^2(\ab)} ( \alpha - \ab) + \frac{ g(\bb) g'(\bb)}{2 f(\ab)} ( \beta - \bb). \label{eq:fun1.line2}
\end{align}
Additionally, by \eqref{eq:Gamh},
\begin{align}
\int_{\Rb^2} \dd \alpha \dd \beta\, \Gamh_0(t,a,b;s,\alpha,\beta)\, \alpha
    &= \Psi_0(t,s) \(a + \mu_0(s -t)\),  \label{eq:gam.alpha}\\
\int_{\Rb^2} \dd \alpha \dd \beta\, \Gamh_0(t,a,b;s,\alpha,\beta)\, \beta
    &=\Psi_0(t,s) \( b + \eta_0 ( s - t)\), \label{eq:gam.beta}
\end{align}
where $\Psi_0$ is given in \eqref{eq:Psi}. Recalling the expression \eqref{eq:h0} for $h_0$ and applying \eqref{eq:fun1.line1}, \eqref{eq:fun1.line2}, \eqref{eq:gam.alpha} and \eqref{eq:gam.beta} to \eqref{eq:h1.pre}, we evaluate the integrals with respect $\alpha$ and $\beta$ present in \eqref{eq:h1.pre} to  obtain
\begin{align}
h_1(t,a,b)
    &=  \int_t^T \dd s\, \Psi_0(t,s)
        \(- \gam^2 f'(\ab)\( a - \ab + \mu_0 (s-t)\)\theta_0^2(s) \right. \\ &\left. \hspace{5cm}
        + \gam g'(\bb) \( b - \bb + \eta_0 ( s- t)\)\theta_0(s)\). \label{eq:h1.int}
\end{align}
We can rewrite the integrand of \eqref{eq:h1.int} as 
\begin{align}
&\Psi_0(t,s)\(- \gam^2 f'(\ab)\( a - \ab + \mu_0 (s-t)\)\theta_0^2(s)  + \gam g'(\bb) \( b - \bb + \eta_0 ( s- t)\)\theta_0(s)\)\\
    = &c^{(1)} \(t,a,b\)  s \theta_0^2(s) \Psi_0(t,s) + c^{(2)} \(t,a,b\) \theta_0^2(s) \Psi_0(t,s) \label{eq:h1.integrand} \\ 
     &\hspace{1cm}+c^{(3)} \(t,a,b\)  s \theta_0(s) \Psi_0(t,s) + c^{(4)} \(t,a,b\) \theta_0(s) \Psi_0(t,s)         
\end{align}
where the $c^{(i)}$ are defined in \eqref{def:h1.summands}. From \eqref{eq:h1.int} and \eqref{eq:h1.integrand}, we see--with the arguments of $c^{(i)}$ suppressed--that
\begin{align}
h_1(t,a,b)
    &= \int_t^T \dd s\, \( c^{(1)} s \theta_0^2(s) \Psi_0(t,s) + c^{(2)} \theta_0^2(s) \Psi_0(t,s) + c^{(3)} s \theta_0(s) \Psi_0(t,s) + c^{(4)} \theta_0(s) \Psi_0(t,s)\) \\
    &= c^{(1)} I^{(1)} \(t\)+ c^{(2)} I^{(2)} \(t\) + c^{(3)} I^{(3)} \(t\) + c^{(4)} I^{(4)} \(t\), 
\end{align}
where the $I^{(i)}$ are defined in \eqref{def:h1.summands}. We have thus established \eqref{eq:h1}. 
\end{proof}

\section{Computation of \texorpdfstring{$h_1$}{h1}, \texorpdfstring{$h_1^{(\infty,\varphi)}$}{h1L} and \texorpdfstring{$h_1^{(\infty,0)}$}{h1LH}}
\label{app:h1}
In Corollary \ref{cor:h1}, we gave a representation of $h_1$ as an integral in time. We begin this section by computing the aforementioned integral, yielding an explicit formula for $h_1$. Afterwards, we take the limits of $h_1$ that correspond to the limiting strategies discussed in Section \ref{sec:limstrats} i.e. we compute $h_1^{(\infty,\varphi)}$ and $h_1^{(\infty,0)}$. 

We now compute the integrals $I^{(i)}$, which are given in \eqref{def:h1.summands}. We have \noeqref{eq:I1} \noeqref{eq:I2} \noeqref{eq:I3} \noeqref{eq:I4}
\begin{subequations}
\label{eqs:I}
\begin{align}
I^{(1)}\(t\)
	&=\int_t^T \dd s\, s \theta_0^2(s) \Psi_0(t,s)\label{eq:I1} \\
	&= \frac{\ee^{2 \gam  (t+T)} \left(4 \gam ^2 \zeta  \left(T^2-t^2\right)-\zeta^2 (2 \gam  T+1)+2 \gam  T-1\right)}{4 \gam ^2 \left(\ee^{2 \gam  t}-\zeta  \ee^{2 \gam  T}\right)^2} \\ & \hspace{3cm} 
	+ \frac{\ee^{4 \gam  t} (1-2 \gam  t)+\zeta ^2 (2 \gam  t+1) \ee^{4 \gam  T}}{4 \gam ^2 \left(\ee^{2 \gam  t}-\zeta  \ee^{2 \gam  T}\right)^2},
	\\
I^{(2)}\(t\)
	&= \int_t^T\dd s\, \theta_0^2(s) \Psi_0(t,s) \label{eq:I2}\\
	&= -\frac{\ee^{4 \gam  t}+\ee^{2 \gam  (t+T)} \left(\zeta ^2+4 \gam  \zeta  (t-T)-1\right)+\zeta ^2 \left(-\ee^{4 \gam  T}\right)}{2 \gam  \left(\ee^{2 \gam  t}-\zeta  \ee^{2 \gam  T}\right)^2},  \\
I^{(3)} \(t\)
	&= \int_t^T \dd s\, s \theta_0(s) \Psi_0(t,s) \label{eq:I3}\\
	&= \frac{\ee^{4 \gam  t} (1-2 \gam  t)+\zeta ^2 (2 \gam  t+1) \left(-\ee^{4 \gam  T}\right)}{4 \gam ^2 \left(\ee^{2 \gam  t}-\zeta  \ee^{2 \gam  T}\right)^2} 
													 \\ &\hspace{3cm}
	+\frac{\ee^{2 \gam  (t+T)} \left(\zeta ^2+2 \gam  \left(\zeta ^2+1\right) T-1\right)}{4 \gam ^2 \left(\ee^{2 \gam  t}-\zeta  \ee^{2 \gam  T}\right)^2}, \\
I^{(4)}\(t\)
	&= \int_t^T \dd s\, \theta_0(s) \Psi_0(t,s)\label{eq:I4} \\
	&= -\frac{\left(\ee^{2 \gam  t}-\ee^{2 \gam  T}\right) \left(\ee^{2 \gam  t}-\zeta ^2 \ee^{2 \gam  T}\right)}{2 \gam  \left(\ee^{2 \gam  t}-\zeta  \ee^{2 \gam  T}\right)^2}.	 
\end{align}
\end{subequations}
Inserting \eqref{eqs:I} into \eqref{eq:h1} gives an explicit expression for $h_1$. 

Let us now compute $h_1^{(\infty,\varphi)}$. To this end, direct computation yields
\begin{subequations}
\label{eqs:I.limits.1}
\begin{align}
I^{(1)}_{(\infty,\varphi)}\(t\)
	&:= \lim_{\kappa \to \infty} I^{(1)}\(t\) 
	=\frac{e^{2 \gam  (t+T)} \left(4 \gam ^2 \left(T^2-t^2\right)-2\right)}{4 \gam ^2 \left(e^{2 \gam  t}-e^{2 \gam  T}\right)^2} 
	+\frac{e^{4 \gam  t} (1-2 \gam  t)+(2 \gam  t+1) e^{4 \gam  T}}{4 \gam ^2 \left(e^{2 \gam  t}-e^{2 \gam  T}\right)^2}, \label{eq:I1L}\\
I^{(2)}_{(\infty,\varphi)}\(t\)
	&:= \lim_{\kappa \to \infty} I^{(2)}\(t\) 
	= \frac{-e^{4 \gam  t}-4 \gam  (t-T) e^{2 \gam  (t+T)}+e^{4 \gam  T}}{2 \gam  \left(e^{2 \gam  t}-e^{2 \gam  T}\right)^2}, \label{eq:I2L} \\
I^{(3)}_{(\infty,\varphi)}\(t\)
	&:= \lim_{\kappa \to \infty} I^{(3)} \(t\) 
	=\frac{e^{4 \gam  t} (1-2 \gam  t)+4 \gam  T e^{2 \gam  (t+T)}-(2 \gam  t+1) e^{4 \gam  T}}{4 \gam ^2 \left(e^{2 \gam  t}-e^{2 \gam  T}\right)^2}, \label{eq:I3L} \\
I^{(4)}_{(\infty,\varphi)}\(t\)
	&:= \lim_{\kappa \to \infty} I^{(4)} \(t\)
	= - \frac 1{2\gam}. \label{eq:I4L}
\noeqref{eq:I1L}\noeqref{eq:I2L}\noeqref{eq:I3L}\noeqref{eq:I4L}
\end{align}
\end{subequations}
We now compute $h_1^{(\infty,0)}$. By direct computation, we obtain \noeqref{eq:I1LH}\noeqref{eq:I2LH}\noeqref{eq:I3LH}\noeqref{eq:I4LH}
\begin{subequations}
\label{eqs:I.limits.2}
\begin{align}
\lim_{\varphi \to 0} c^{(1)}(t,a,b) I^{(1)}_{(\infty,\varphi)} (t)
    &= - \frac{ f'(\ab) \mu_0(T+t)}{2 (T-t)},                       \label{eq:I1LH}\\
\lim_{\varphi \to 0} c^{(2)}(t,a,b) I^{(2)}_{(\infty,\varphi)} (t) 
    &= - \frac{f'(\ab)(a - \ab - t \mu_0) }{T-t},                   \label{eq:I2LH} \\
\lim_{\varphi \to 0} c^{(3)}(t,a,b) I^{(3)}_{(\infty,\varphi)}(t) 
    &= - \frac 16 g'(\bb) \eta_0 (T + 2t),                          \label{eq:I3LH}\\
\lim_{\varphi \to 0} c^{(4)}(t,a,b) I^{(4)}_{(\infty,\varphi)} (t) 
    &= - \frac 12 g'(\bb) (b - \bb - t \eta_0),                     \label{eq:I4LH}
\end{align}
\end{subequations}
where the $c^{(i)}$ are given in \eqref{def:h1.summands} and $I^{(i)}_{(\infty,\varphi)}$ are given in \eqref{eqs:I.limits.1}.  Summing the terms in \eqref{eqs:I.limits.2} yields $h_1^{(\infty,0)}$, which is given in \eqref{eq:hLH}.

%
%

\clearpage
\section{Figures}

%
%

\begin{figure}[h]
\centering
\label{fig:ex1}
    \begin{subfigure}[b]{0.47\textwidth}
        \includegraphics[width=\textwidth]{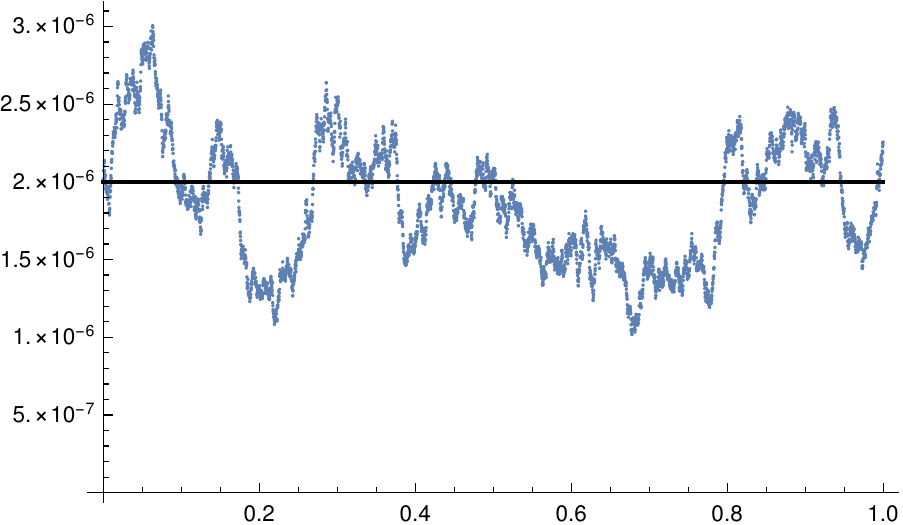}
        \caption{}
        \label{fig:tempi}
    \end{subfigure}
    \hfill
    \begin{subfigure}[b]{0.47\textwidth}
        \includegraphics[width=\textwidth]{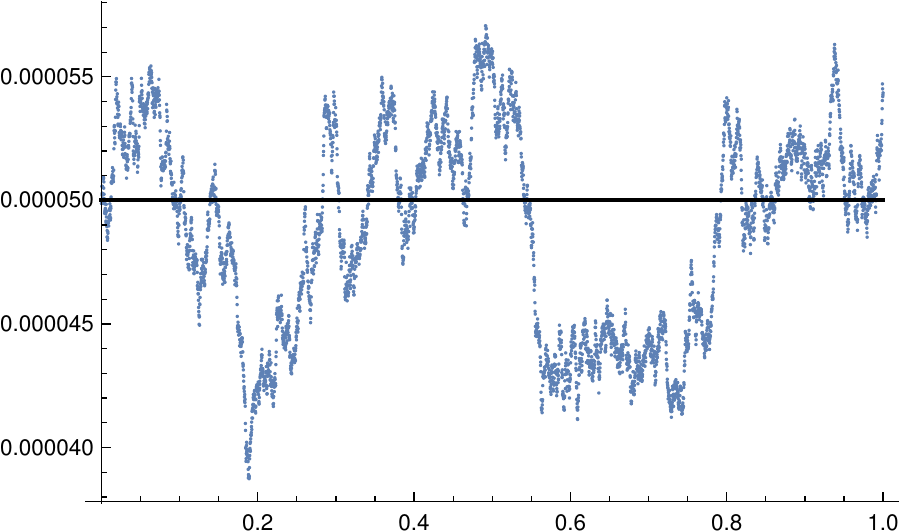}
        \caption{}
        \label{fig:permi}
    \end{subfigure}

    \begin{subfigure}[b]{0.47\textwidth}
        \includegraphics[width=\textwidth]{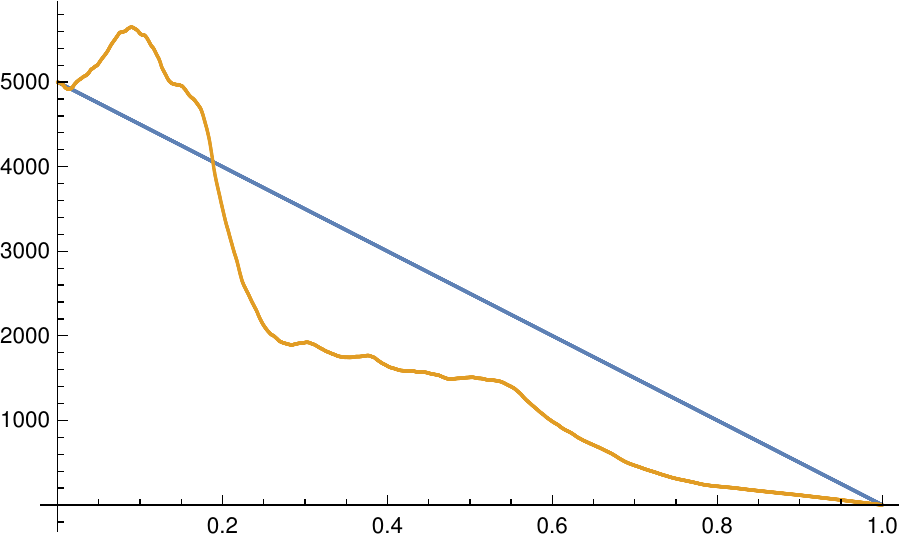}
        \caption{}
        \label{fig:invi}
    \end{subfigure}
    \hfill
    \begin{subfigure}[b]{0.47\textwidth}
        \includegraphics[width=\textwidth]{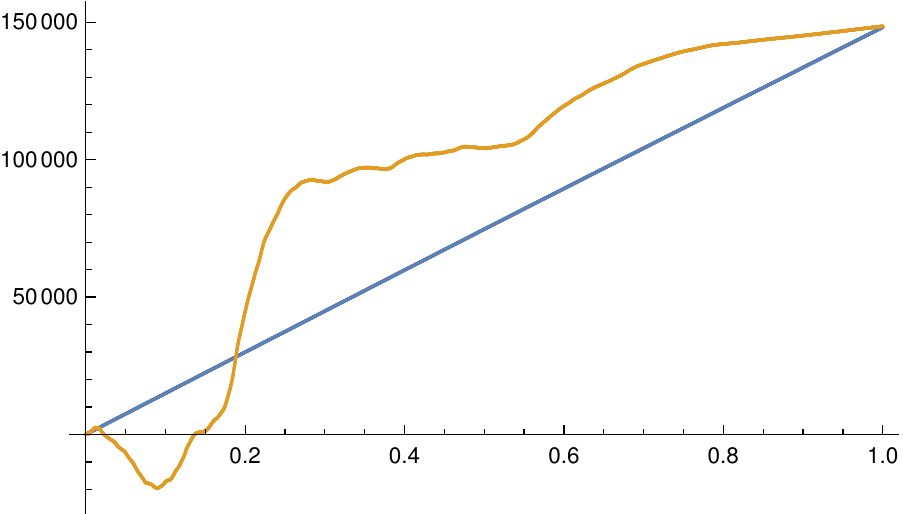}
        \caption{}
        \label{fig:cashi}
    \end{subfigure}
    \caption*{}
\vspace{-1.1cm}
\end{figure}

\begin{wrapfigure}{l}{0.53\textwidth}
	\addtocounter{figure}{-1}
	\refstepcounter{figure}
	\label{fig:ex1.2}
	\addtocounter{figure}{-1}
    \begin{subfigure}{.47\textwidth}
        \addtocounter{subfigure}{4}
        \includegraphics[width=\textwidth]{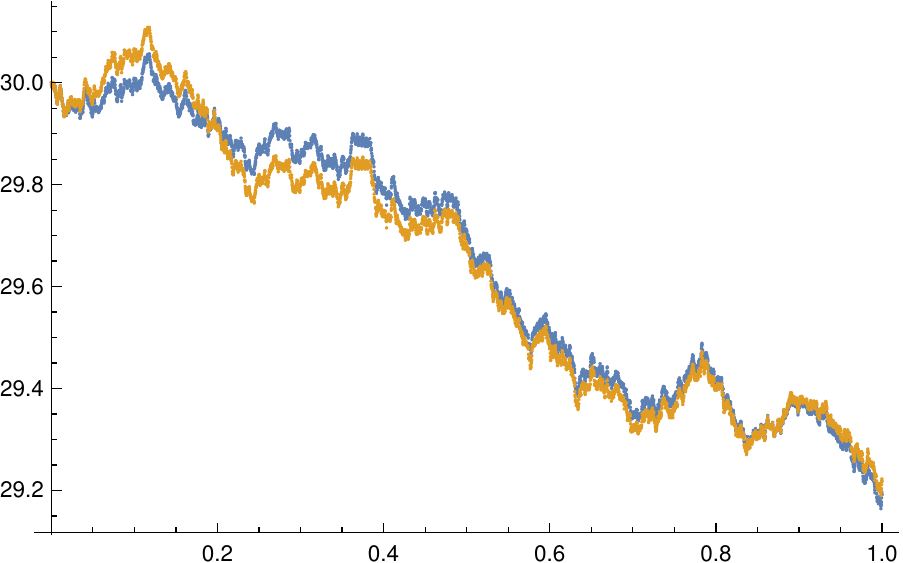}
        \caption{}
        \label{fig:stocki}
    \end{subfigure}
    \vspace{-1.1cm}
\end{wrapfigure}
\noindent
Figure \thefigure: Here we plot a single sample path of $(a,b)$ and the paths of $(X^\nu,S^\nu,Q^\nu)$ that result from following $\nu = \nuba 0^{(\infty,0)}$ (\blu{blue}) and $\nu = \nuba 1^{(\infty,0)}$ (\ora{orange}) with dynamics \eqref{eq:impact.dynamics.ex1} and parameters \eqref{eq:ex1.param.values}. In Figure \ref{fig:tempi}, we plot the temporary price impact $a$, and in Figure \ref{fig:permi} we plot the permanent price impact $b$. We plot the trader's inventory $Q^\nu$ in Figure \ref{fig:invi} and the trader's cash $X^\nu$ in Figure \ref{fig:cashi}. In Figure \ref{fig:stocki}, we plot the stock prices $S^\nu$.

%
%

\begin{figure}[h]
    \centering
    \begin{subfigure}[b]{0.47\textwidth}
        \includegraphics[width=\textwidth]{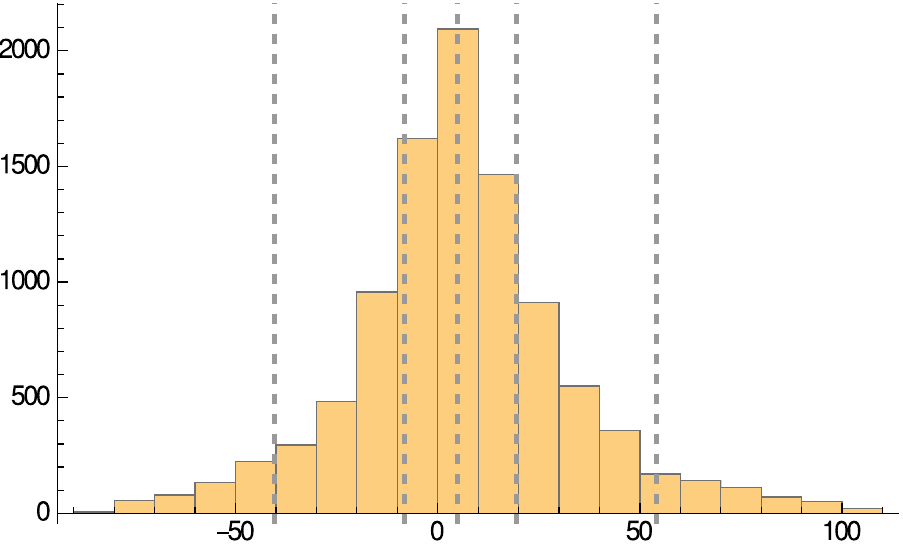}
         \caption{\label{fig:trim-centered-zero-over-ac}}
    \end{subfigure}
    \hfill
    \begin{subfigure}[b]{0.47\textwidth}
        \includegraphics[width=\textwidth]{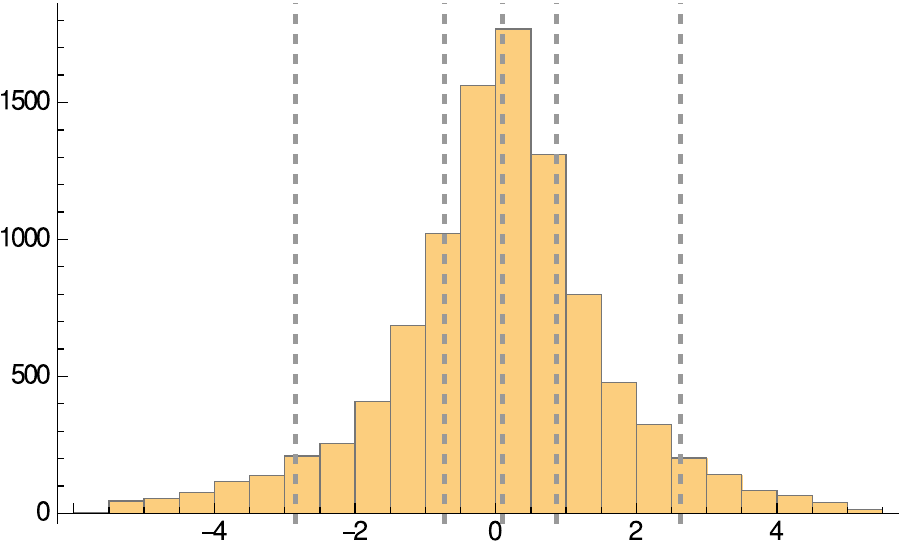}
         \caption{\label{fig:trim-centered-first-over-zero}}
    \end{subfigure}
    \caption{Here we plot histograms of the relative performance criteria given in \eqref{eq:rel.perf} with the initial conditions \eqref{eq:init.cond}. In Figure \ref{fig:trim-centered-zero-over-ac} we plot the performance of $\nuba 0$ relative to $\nu_{AC}$. In Figure \ref{fig:trim-centered-first-over-zero}, we plot performance of $\nuba 1$ relative to $\nuba 0$. The vertical, dashed lines represent the 5\%, 25\%, 50\%, 75\%, and 95\% quantiles, respectively.}
    \label{fig:trim-centered-base}
\end{figure}

\begin{figure}[h]
    \centering
    \begin{subfigure}[b]{0.47\textwidth}
        \includegraphics[width=\textwidth]{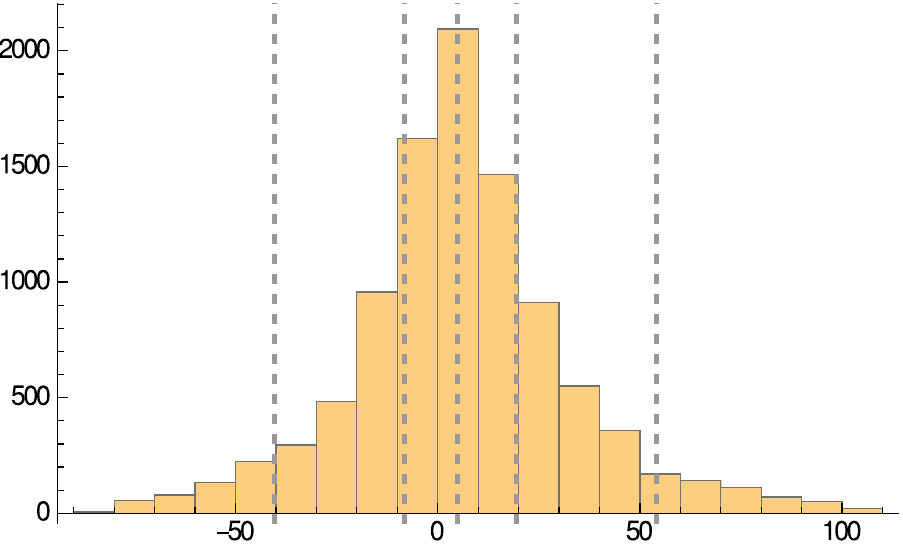}
         \caption{}
         \label{fig:trim-centered-zero-over-ac-L}
    \end{subfigure}
    \hfill
    \begin{subfigure}[b]{0.47\textwidth}
        \includegraphics[width=\textwidth]{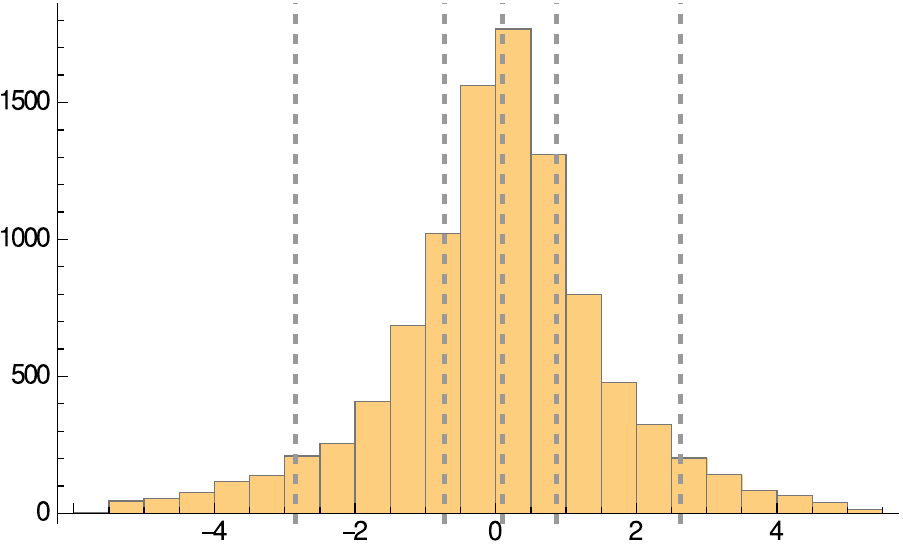}
         \caption{}
         \label{fig:trim-centered-first-over-zero-L}
    \end{subfigure}
    \caption{Here we plot histograms of the relative performance criteria given in \eqref{eqs:rel.perf.L} with the initial conditions \eqref{eq:init.cond}. In Figure \ref{fig:trim-centered-zero-over-ac-L}, we plot the performance \eqref{eq:rel.perf.L.0} of $\nuba 0^{(\infty,\varphi)}$ relative to $\nu_{AC}^{(\infty,\varphi)}$. In Figure \ref{fig:trim-centered-first-over-zero-L}, we plot performance \eqref{eq:rel.perf.L.1} of $\nuba 1^{(\infty,\varphi)}$ relative to $\nuba 0^{(\infty,\varphi)}$. The vertical, dashed lines represent the 5\%, 25\%, 50\%, 75\%, and 95\% quantiles, respectively.}
    \label{fig:trim-centered-L}
\end{figure}

\begin{figure}[ht]
    \begin{minipage}[t]{.45\textwidth}
        \centering
        \includegraphics[width=\textwidth]{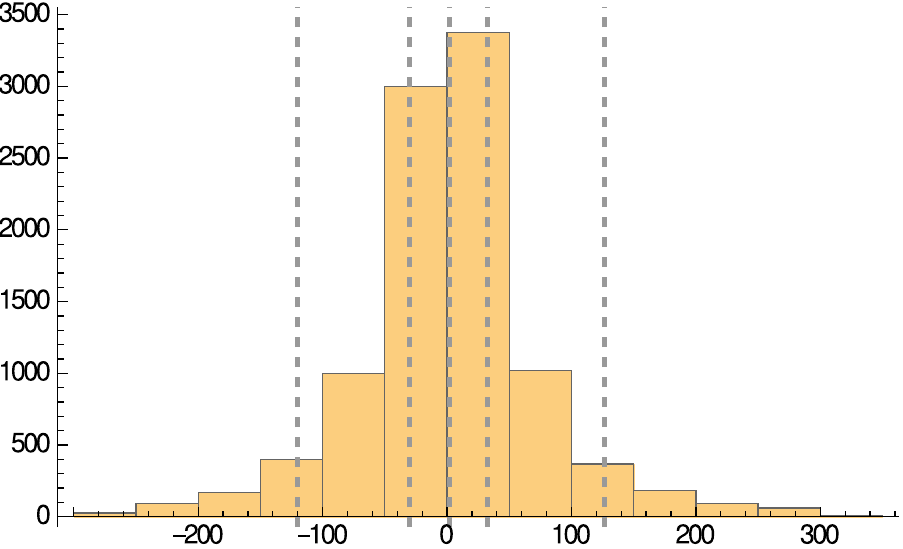}
        \caption{Here we plot the performance of $\nuba 1^{(\infty,0)}$ relative to $\nuba 0^{(\infty,0)}$ with respect to the performance criteria \eqref{eq:rel.perf.LH} with initial conditions \eqref{eq:init.cond}.  The vertical, dashed lines represent the 5\%, 25\%, 50\%, 75\%, and 95\% quantiles, respectively. }
		\label{fig:trim-centered-first-over-zero-LP}
    \end{minipage}
    \hfill
    \begin{minipage}[t]{.45\textwidth}
        \centering
        \includegraphics[width=\textwidth]{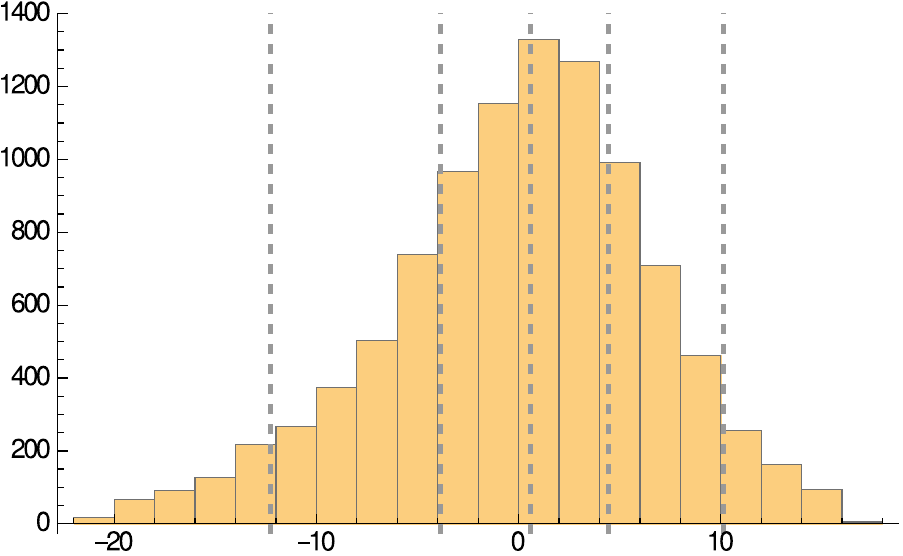}
        \caption{Here we plot the performance of $\nuba 1$ relative to $\nuba{0}$ with respect to the performance criteria \eqref{eq:rel.perf} (right) with initial conditions \eqref{eq:init.cond.above}. The vertical, dashed lines represent the 5\%, 25\%, 50\%, 75\%, and 95\% quantiles, respectively. }
        \label{fig:trim-above-first-over-zero}
    \end{minipage}
\end{figure}

\begin{figure}[ht]
    \begin{minipage}[t]{.45\textwidth}
        \centering
        \includegraphics[width=\textwidth]{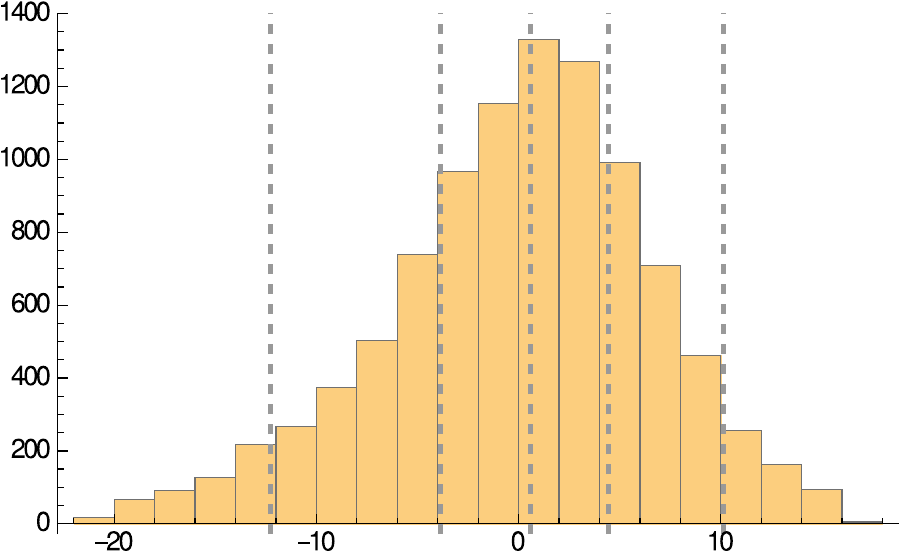}
        \caption{Here we plot the performance of $\nuba 1^{(\infty,\varphi)}$ relative to $\nuba{0}^{(\infty,\varphi)}$ with respect to the performance criteria \eqref{eq:rel.perf.L.1} with initial conditions \eqref{eq:init.cond.above}. The vertical, dashed lines represent the 5\%, 25\%, 50\%, 75\%, and 95\% quantiles, respectively. }
        \label{fig:trim-above-first-over-zero-L}
    \end{minipage}
    \hfill
    \begin{minipage}[t]{.45\textwidth}
        \centering
        \includegraphics[width=\textwidth]{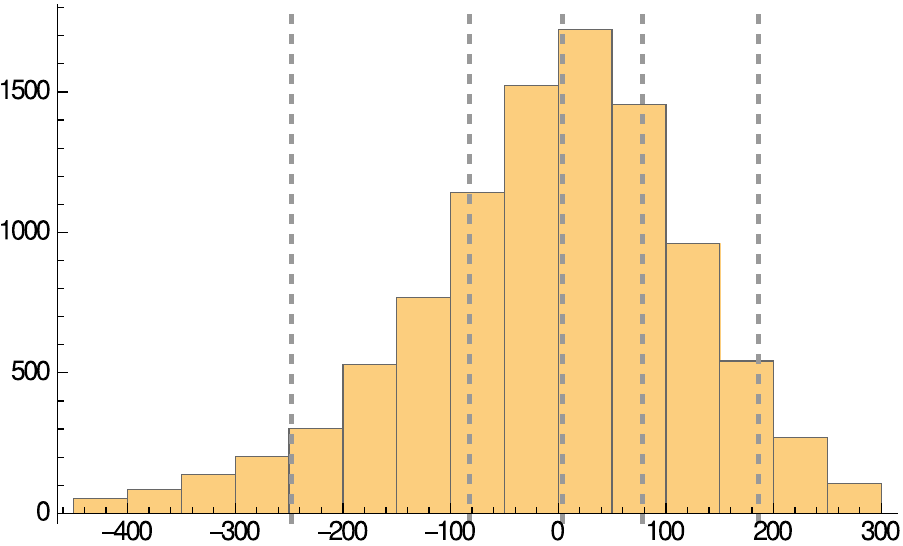}
        \caption{Here we plot the performance of $\nuba 1^{(\infty,0)}$ relative to $\nuba{0}^{(\infty,0)}$ with respect to the performance criteria \eqref{eq:rel.perf.LH} with initial conditions \eqref{eq:init.cond.above}. The vertical, dashed lines represent the 5\%, 25\%, 50\%, 75\%, and 95\% quantiles, respectively. }
        \label{fig:trim-above-first-over-zero-LH}
    \end{minipage}
\end{figure}

\end{document}